\documentclass[aps,prd,
eqsecnum,nofootinbib,preprintnumbers]{revtex4}

\usepackage{amsmath,amsfonts,amssymb,color,latexsym,amsthm,mathrsfs,comment,dsfont
}
\usepackage{hyperref,graphicx}

\newtheorem{lemma}{Lemma}
\newcommand{\ma}[1]{\mbox{$\mathcal{#1}$}}
\newcommand{\mas}[1]{\mbox{$\mathscr{#1}$}}

\newcommand{\D}{{\rm d}}

\newcommand{\ti}{\tilde}

\begin{document}

\title{
Exact spinning wormholes in five dimensions from Toda integrability
}

\author{Masato Nozawa}
\email{
masato.nozawa@oit.ac.jp
}

\address{
Department of General Education, Faculty of Engineering, 
Osaka Institute of Technology, 5-16-1, Omiya, Asahi-ku, 
Osaka, Osaka 535-8585, Japan
}

\date{\today}

\begin{abstract} 
We construct a new exact solution to the five-dimensional Einstein-phantom-scalar system that 
describes a rotating wormhole connecting two asymptotically flat regions. 
The solution represents a rotating generalization of the Ellis-Bronnikov wormhole with equal angular momenta.
Under the cohomogeneity-one ansatz, we demonstrate that the radial equations reduce to an integrable $A_1\oplus A_2$ Toda system, from which the general local solution can be obtained in closed form. This integrable structure enables us to implement the regularity and asymptotic-flatness conditions analytically, thereby singling out a three-parameter family of solutions characterized by the mass, the equal angular momenta, and the scalar charge. The resulting spacetimes are geodesically complete, horizonless and stably causal, and contain no curvature singularities. The Toda structure further allows us to prove the uniqueness of the wormhole throat 
and to derive a sharp upper bound on the angular momentum for a fixed throat area. 
We also establish the existence and uniqueness of an ergosurface with respect to the stationary frame that is nonrotating at one asymptotic end, while the opposite end rotates relative to this frame. 
In a suitable scaling limit, the scalar field vanishes and the solution approaches the exterior region of
the extremal Myers-Perry black hole with equal angular momenta. 
\end{abstract}


\maketitle


\section{Introduction}

Exact solutions to the Einstein equations have played an indispensable role in understanding the nonlinear nature of gravity. Their construction, however, is notoriously difficult, and the adversity faced in this endeavor becomes particularly severe when rotation is present~\cite{Stephani:2003tm}. In higher-dimensional vacuum gravity, considerable progress has nevertheless been achieved by exploiting hidden integrable structures of stationary spacetimes. In five dimensions, for example, solution-generating techniques such as the inverse-scattering method have led to a rich variety of rotating black objects with no four-dimensional counterparts; see, e.g., \cite{Emparan:2008eg,Iguchi:2011qi}.

The situation becomes considerably more involved in the presence of matter fields. Although certain special theories, including supergravity models, retain sufficiently large hidden symmetries after dimensional reduction and admit powerful coset-based solution-generating techniques~\cite{Breitenlohner:1987dg}, such structures are not available for generic matter couplings. Scalar fields provide a particularly simple example for which the field equations remain highly nonlinear once rotation is introduced.

The same difficulty arises in the study of wormholes. Wormholes provide nontrivial bridges connecting otherwise distant regions of spacetime. Within general relativity, traversable wormholes typically require a violation of the null energy condition~\cite{Morris:1988cz,Morris:1988tu,Teo:1998dp}, which is often modeled by exotic matter fields. Whereas static scalar-supported wormholes, such as the Ellis-Bronnikov solution~\cite{Ellis:1973yv,Bronnikov:1973fh}, are known in closed form, obtaining genuinely rotating counterparts analytically is a substantially harder problem. Consequently, most fully nonlinear studies of rotating phantom-scalar wormholes have relied on numerical methods~\cite{Kleihaus:2014dla,Chew:2016epf,Chew:2019lsa}. For instance, the analytic solution presented in~\cite{Matos:2005uh} covers only an interior region and does not provide a global wormhole geometry connecting two asymptotically flat ends.

A natural way to render the rotating problem more tractable is to restrict attention to geometries of cohomogeneity one, for which the field equations reduce to ordinary differential equations. In four dimensions, the Newman-Unti-Tamburino (NUT) deformation provides a familiar example in which a gravitomagnetic deformation can be incorporated while retaining cohomogeneity one~\cite{Klemm:2015uba}. The resulting wormhole spacetime~\cite{Nozawa:2026jls}, however, is not asymptotically flat in the usual sense. In contrast, five dimensions offer a qualitatively different possibility. For the doubly rotating Myers-Perry geometry~\cite{Myers:1986un}, setting the two angular momenta equal enhances the rotational symmetry, allowing a genuinely rotating and asymptotically flat spacetime to maintain cohomogeneity one~\cite{Gibbons:2004uw}. The equal-angular-momenta sector therefore provides a particularly attractive setting in which to search for exact rotating wormholes.

Nevertheless, no closed-form solution has been known for rotating five-dimensional Einstein-phantom-scalar wormholes with equal angular momenta. Previous numerical studies~\cite{Dzhunushaliev:2013jja,Uemichi:2026dzb} have demonstrated the existence of regular rotating configurations and revealed several intriguing properties, including an upper bound on the angular momentum and a limiting relation to the extremal Myers-Perry geometry. An exact solution would therefore be invaluable for establishing their global structure analytically and for clarifying which of these properties follow from the underlying structure of the field equations. 
More importantly, an exact treatment may expose an integrable structure that is invisible in a purely numerical formulation.

In the present paper, we show that such an exact construction is indeed possible for the five-dimensional Einstein system coupled to a phantom scalar field.  Starting from the most general stationary cohomogeneity-one ansatz with equal angular momenta, we reduce the field equations to a one-dimensional mechanical system.  The key observation is that, after an appropriate choice of variables and gauge, the radial equations decompose into a Liouville equation and an $A_2$ Toda system.  Equivalently, the system possesses an $A_1\oplus A_2$ Toda structure, with the two sectors coupled only through the Hamiltonian constraint.  This integrability allows the local solution to be obtained in full generality.  Imposing asymptotic flatness and regularity then selects a three-parameter family of smooth rotating wormholes, which may be parametrized by the mass, the equal angular momenta, and the scalar charge.

The availability of the exact solution enables us to establish analytically a number of global properties reported in ~\cite{Dzhunushaliev:2013jja,Uemichi:2026dzb}.  We show that the spacetime is free of Killing horizons and curvature singularities, is stably causal, and is geodesically complete.  The Toda structure further provides an efficient way to prove the existence and uniqueness of the wormhole throat and to derive an upper bound on the angular momentum for a fixed throat area.  We also analyze the ergoregion, showing that the ergosurface defined with respect to a nonrotating frame at one asymptotic end exists uniquely.  The two asymptotic regions are in general rotating relative to each other, and their conserved quantities obey a simple relation following from the Komar integrals.

The exact solution also admits several limiting geometries of physical interest.  In the zero-angular-momentum limit it reduces to the five-dimensional Ellis-Bronnikov wormhole, whereas a distinct scaling limit approaches the exterior region of the extremal Myers-Perry black hole with equal angular momenta.  Another scaling limit yields a rotating generalization of the Ellis-Gibbons solution.  These limits provide useful probes of the boundary of the parameter space and, in the Toda description, correspond to different degenerations of the underlying integrable system.

Finally, we examine whether the Toda structure found here persists in higher odd dimensions.  Although the equal-angular-momenta cohomogeneity-one ansatz continues to exist, the reduced equations no longer organize themselves into a Cartan-type Toda system in the same manner.  In this sense, the $A_1\oplus A_2$ integrability uncovered in the present work is a special feature of five dimensions.  The solution therefore provides not only a rare exact example of a rotating scalar-supported wormhole, but also a setting in which the geometric properties of the spacetime can be traced directly to an underlying integrable structure.

The remainder of this paper is organized as follows. In section~\ref{sec:pre},
we fix our notation and review some basic properties of asymptotically flat
wormholes. In section~\ref{sec:construction}, we reduce the Lagrangian to
that of a one-dimensional Toda system and construct an exact solution.
The physical properties of the wormhole solution are explored in
section~\ref{sec:phys}. Concluding remarks and future directions are
presented in section~\ref{sec:conc}. The appendixes review the Toda system and derive the Myers-Perry
solution, present curvature invariants, describe a scaling limit
to the rotating Ellis-Gibbons solution, and examine extensions
of the construction to higher dimensions. 

\bigskip
{\it Note added.} 
After the submission of this paper to arXiv, Ref.~\cite{Blazquez-Salcedo:2026wbi} appeared, presenting an
independent construction of the same family of rotating wormhole solutions. Their work focuses on the radial instability of the reflection-symmetric family, complementing our analysis of the global geometry.


\section{Preliminaries}
\label{sec:pre}

In this paper, we consider five-dimensional Einstein gravity sourced by a massless scalar field, which is described by 
\begin{align}
\label{Lag5}
\ma L_5 =R -\epsilon (\nabla \phi)^2 \,,
\end{align}
where $\epsilon=+1$ for an ordinary scalar field, whereas $\epsilon=-1 $ for a phantom field. 
Although we are ultimately interested in the phantom case, we retain 
$\epsilon$ for the moment in order to keep track of the sign of the scalar kinetic term.
Einstein's equations and the scalar field equations are 
\begin{align}
\label{Eineq}
R_{\mu\nu}=\epsilon \nabla_\mu \phi \nabla_\nu \phi \,, \qquad \nabla^2 \phi=0 \,.   
\end{align}
 
Since we are mainly interested in asymptotically flat solutions, 
we first summarize our conventions for asymptotically flat
five-dimensional spacetimes and their conserved charges. 
From the linearized analysis, 
the metric can be expanded asymptotically $r\to \infty$ as \cite{Emparan:2008eg}
\begin{align}
\label{asymet1}
\D s^2\simeq&\,-\left(1-\frac{\mathsf m}{r^2}\right)\D t^2+\left(1+\frac{\mathsf m}{2r^2}\right)\left[
\D r^2+r^2 (\D \theta^2 +\sin^2\theta \D \phi_1^2+\cos^2\theta\D \phi_2^2)
\right]\notag \\
&
- \frac{2\mathsf m}{r^2}\D t \left(\mathsf a \sin^2\theta \D \phi_1+\mathsf b \cos^2\theta \D \phi_2\right)\,.
\end{align}
Here, $\mathsf m$, $\mathsf a$ and $\mathsf b$ are constants related to the ADM mass and angular momenta 
in each orthogonal plane as 
\begin{align}
\label{}
M=\frac{3\pi \mathsf m}{8G} \,, \qquad J_1=\frac{2}{3} M\mathsf a \,, \qquad 
J_2=\frac 23 M\mathsf b \,.   
\end{align}
Instead of the coordinates ($\theta$, $\phi_1$, $\phi_2$) parametrizing $S^3$ 
with the range $0\le \theta\le\pi/2$, $0\le \phi_{1,2}\le 2\pi$, 
it is useful to introduce Euler angles ($\vartheta, \psi, \varphi$) by 
\begin{align}
\label{}
\vartheta = 2 \theta \,, \qquad \psi=\phi_2+\phi_1 \,, \qquad 
\varphi=\phi_2-\phi_1 \,,
\end{align}
where $0\le \psi \le 4\pi$, $0\le \vartheta \le \pi$ and $0\le \varphi \le 2\pi$. 
In terms of these coordinates, the asymptotic metric can be rewritten as
\begin{align}
\label{asymet2}
\D s^2\simeq &-\left(1-\frac{\mathsf m}{r^2}\right)\D t^2+\left(1+\frac{\mathsf m}{2r^2}\right)\left[
\D r^2+\frac 14r^2 (\sigma_1^2+\sigma_2^2+\sigma_3^2)
\right]
- \frac{\mathsf m}{2r^2}\D t \left[(\mathsf a+\mathsf b)\sigma_3-(\mathsf a-\mathsf b)\sigma_3'\right]\,, 
\end{align}
where 
$\{\sigma_1, \sigma_2, \sigma_3\}$ denote the ${\rm SU}(2)$ left-invariant 1-forms 
\begin{align}
\label{}
\sigma_1+i \sigma_2 \equiv e^{i\psi}\left(i \D \vartheta +\sin\vartheta\D \varphi\right)\,, \qquad 
\sigma_3 \equiv \D \psi+\cos\vartheta \D \varphi \,, 
\end{align}
with $\sigma_3'=\D \varphi+\cos\vartheta \D \psi$. 
In the special case of equal angular momenta $\mathsf a=\mathsf b$, the cross-term involving $\sigma_3'$ vanishes, and the asymptotic metric simplifies to
\begin{align}
\label{}
\D s^2\simeq -\left(1-\frac{\mathsf m}{r^2}\right)\D t^2+\left(1+\frac{\mathsf m}{2r^2}\right)\left[
\D r^2+\frac 14r^2 (\sigma_1^2+\sigma_2^2+\sigma_3^2)
\right]
- \frac{\mathsf m\mathsf a}{r^2}\D t\sigma_3\,, 
\end{align}
with the corresponding angular momenta given by
\begin{align}
\label{}
J^{(\psi)}=\frac 12 (J_1+J_2)=\frac 23 M\mathsf a \,, \qquad J^{(\varphi)}=\frac 12 (J_2-J_1)=0 \,. 
\end{align}

\subsection{Static wormhole}
\label{sec:static}

Before going into the details of the rotating solution, it is illustrative to 
review its static counterpart.
A higher-dimensional generalization of the Ellis-Bronnikov wormhole~\cite{Ellis:1973yv,Bronnikov:1973fh}
 was constructed in \cite{Torii:2013xba} and further generalized in~\cite{Dzhunushaliev:2013jja}. Following the notation in~\cite{Martinez:2020hjm}, 
the five-dimensional solution reads 
\begin{align}
\label{EB}
\D s^2= &\,-e^{-2\beta u}\D t^2+\frac{me^{\beta u}}{\sin u}\left[\frac{\D u^2}{4\sin^2 u}+\frac 14 \left(\sigma_1^2+\sigma_2^2+\sigma_3^2\right)\right]\,, 
\\
\phi=&\, \pm \sqrt{-\frac 32 \epsilon (1+\beta^2)}\, u \,,
\label{EBphi}
\end{align}
where $\beta \in\mathbb R$ and $m(>0)$ are constants. 
For the reality of the scalar field, equation (\ref{EBphi}) requires $\epsilon=-1$, viz., 
the supporting matter field is necessarily phantom. 
Since the metric $(\sigma_1^2+\sigma_2^2+\sigma_3^2)/4$ represents 
a unit three-sphere, the solution admits a spatial ${\rm SO}(4)\simeq {\rm SU}(2)\times {\rm SU}(2)$ symmetry. 
The coordinate $u$ ranges over $0< u <\pi$, with $u=0$ and $u=\pi$
being two asymptotic ends. To make the asymptotic flatness manifest, it is useful to introduce 
a coordinate $x$ by $u=U(x)$, where 
\begin{align}
\label{U}
U(x)\equiv \frac \pi 2 -\arctan \left(\frac{x}{m}\right)\,.
\end{align}
The metric is now reduced to 
\begin{align}
\label{EBx}
\D s^2= &\,-e^{-2\beta U(x)}\D t^2+e^{\beta U(x)}\left[\frac{\D x^2}{4\sqrt{x^2+m^2}}+\frac 14 
\sqrt{x^2+m^2}(\sigma_1^2+\sigma_2^2+\sigma_3^2)\right]\,, 
\end{align}
Around one asymptotic end $r\equiv x^{1/2}\to +\infty$, 
the solution behaves as 
\begin{align}
\label{}
\D s^2 \simeq &\, -\left(1-\frac{2m \beta}{r^2}\right)\D t^2+\left(1+\frac{m \beta}{r^2}\right)
\left[\D r^2+\frac{r^2}{4}(\sigma_1^2+\sigma_2^2+\sigma_3^2)\right]\,, \\
\phi\simeq &\,\pm \sqrt{-\frac{3}{2}\epsilon (1+\beta^2)}\frac{m}{r^2 } \,.
\end{align}
Around the other asymptotic end $x\to -\infty$, the solution is expanded as 
\begin{align}
\label{}
\D s^2\simeq &\,-e^{-2\pi \beta}\left(1+\frac{2m\beta}{|x|}\right)\D t^2
+\frac{e^{\pi \beta}}{4|x|} \left(1-\frac{m\beta}{|x|}\right)\D x^2
+\frac {|x|e^{\pi \beta}}{4}\left(1-\frac{m\beta}{|x|}\right)(\sigma_1^2+\sigma_2^2+\sigma_3^2)
\notag \\
=&\, -\left(1+\frac{2m\beta e^{\pi \beta}}{\ti r^2}\right)\D \ti t\,^2
+\left(1-\frac{m\beta e^{\pi \beta}}{\ti r^2}\right)\left[\D\ti r^2+\frac {\ti r^2}{4}(\sigma_1^2+\sigma_2^2+\sigma_3^2)\right]\,,
\\
\phi\simeq &\, \pm \sqrt{-\frac 32 \epsilon (1+\beta^2)}\left(\pi-\frac{me^{\pi\beta}}{\ti r^2}\right)\,, 
\end{align}
where we have defined $\ti t=e^{-\pi \beta} t$ and $\ti r=e^{\pi \beta/2}(-x)^{1/2}$
in the second equality. 
Thus, each end is asymptotically flat in its own canonically normalized
coordinates, and the corresponding ADM masses are
\begin{align}
\label{}
M_+=\frac{3\pi}{4 G} \beta m \,, \qquad 
M_-=-\frac{3\pi}{4G} \beta  e^{\beta \pi } m\,. 
\end{align}
Note that $M_+$ and $M_-$ inevitably have opposite signs for $\beta \ne 0$. 

The wormhole throat is located at $x=m\beta $, where 
the areal radius $e^{\beta U/2}(x^2+m^2)^{1/4}$ takes a minimum. 
Since $e^{-2\beta U(x)}$ is strictly positive everywhere, the spacetime contains no Killing horizon.
Moreover, it is free of curvature singularities and is geodesically complete.
These properties constitute the basic geometric requirements for a regular traversable wormhole.

The Ellis-Bronnikov solution has been extended in several directions: dilatonic charged solutions \cite{Huang:2019arj,Nozawa:2020wet}, asymptotically anti-de Sitter solutions \cite{Huang:2020qmn,Nozawa:2020gzz,Blazquez-Salcedo:2020nsa}. In $D=4$ dimensions, 
further extensions to the C-metric \cite{Nozawa:2023aep}, the Robinson-Trautman solution
\cite{Tahamtan:2016fur,Nozawa:2023boa} and the NUT-charged solution \cite{Nozawa:2026jls} are possible.

\section{Construction of the solution}
\label{sec:construction}

\subsection{Cohomogeneity-one reduction}

We would like to  seek a rotating extension of the wormhole solution (\ref{EB}). 
To simplify the analysis, 
we adopt the following stationary cohomogeneity-one ansatz
\begin{align}
\label{metric0}
\D s^2 =- e^{2A(z)}\D t^2+N(z)^2 \D z^2+\frac 14 e^{2B(z)}(\sigma_1^2+\sigma_2^2)
+\frac 14 e^{2C(z)}\left(\sigma_3-\omega(z)\D t\right)^2 \,,\qquad 
\phi=\phi(z) \,. 
\end{align}
Accordingly, the solution depends nontrivially on a single coordinate $z$. 
Here, we leave the lapse function $N(z)$ unfixed, since we work in the 
Lagrangian framework, in which it plays the role of the Lagrange multiplier. 
Note that the above ansatz preserves the ${\rm U}(2)\simeq {\rm U}(1)\times {\rm SU}(2)$
spatial symmetry, in addition to the time translation symmetry. 
For the numerical analysis of this system, see \cite{Dzhunushaliev:2013jja,Uemichi:2026dzb}.

Inserting (\ref{metric0}) into (\ref{Lag5}) and discarding the total divergence, one obtains the effective one-dimensional Lagrangian
\begin{align}
\label{Lag1D}
L=\frac {e^{A+2B+C}}{N}\left[2 A'(2B'+C')+2B'(B'+2C')
- \epsilon \phi'{}^2 +\frac 1{8}e^{2(C-A)}\omega'{}^2 \right]
+2N  e^{A+2B+C}\left(4e^{-2B}-e^{2(C-2B)}\right)\,,
\end{align}
where a prime denotes differentiation with respect to $z$. 
Note that the last term in (\ref{Lag1D}) descends from the 
scalar curvature $R^{(3)}=2 (4e^{-2B}-e^{2(C-2B)})$  of the Berger sphere 
\begin{align}
\label{}
\D s_{\rm Berg}^2 = \frac 14 e^{2B(z)}(\sigma_1^2+\sigma_2^2)
+\frac 14 e^{2C(z)}\sigma_3^2 \,. 
\end{align} 
Variation of the Lagrangian (\ref{Lag1D}) with respect to $(A,B,C,\omega,\phi)$ gives rise to 
the corresponding Euler-Lagrange equations, while 
variation with respect to the lapse $N$ yields the Hamiltonian constraint.

Since $\omega$ and $\phi$ are cyclic variables,  their conjugate momenta are conserved
\begin{align}
\label{cyclic}
p_\omega \equiv \frac{e^{-A+2B+3C}\omega'}{4N} =j\,, \qquad 
p_\phi \equiv -\frac{2\epsilon e^{A+2B+C}\phi'}{N}=-2 Q\,,
\end{align}
where $j$ and $Q$ are constants. 
The Routh reduction $\ti L=L-(p_\omega \omega'+p_\phi \phi')$ 
gives rise to a reduced Lagrangian 
\begin{align}
\label{}
\tilde L=&\, \frac {2e^{A+2B+C}}{N}\left[ A'(2B'+C')+B'(B'+2C')\right]
+2N  e^{A+2B+C}\left(4e^{-2B}-e^{2(C-2B)}\right)
\notag \\
&+N \left(-2 e^{A-2B-3C}j^2 + \epsilon e^{-A-2B-C}Q^2 \right)\,.
\label{Lag1D2}
\end{align}
The corresponding Hamiltonian constraint  reads
\begin{align}
\label{}
\ti H =0 \,, 
\end{align}
where 
\begin{align}
\label{}
\ti H\equiv &\, \frac {2e^{A+2B+C}}{N}\left[ A'(2B'+C')+B'(B'+2C')\right]
-2N  e^{A+2B+C}\left(4e^{-2B}-e^{2(C-2B)}\right)
\notag \\
&+N \left(2 e^{A-2B-3C}j^2 -\epsilon e^{-A-2B-C}Q^2 \right)\,.
\end{align}
Henceforth we consider the case where $j$ is nonvanishing. 
By flipping the direction of rotation if necessary, one can assume $j>0$
without losing any generality.

To make the analysis more manageable, define new variables $(q_0(z), q_1(z), q_2(z))$ by
\begin{align}
\label{}
q_0 \equiv  A+B+C+\ln 2 \,, \qquad 
q_1 \equiv 2(A+C)+\frac 13 \ln (64j^2) \,, \qquad 
q_2 \equiv 2 A+\frac 23 \ln (8j^2) \,,
\end{align}
and work in a gauge 
\begin{align}
\label{gauge}
N=e^{A+2B+C} =\frac 12 e^{2q_0-\frac 12 q_1}j^{1/3}  \,. 
\end{align}
Then, the relation (\ref{cyclic}) gives 
\begin{align}
\label{omegaphip}
\omega '=\frac{e^{2q_2-q_1}}{j}\,, \qquad \phi'= \epsilon  Q \,, 
\end{align}
and the reduced Lagrangian (\ref{Lag1D2}) simplifies considerably to 
\begin{align}
\label{Lq}
\tilde L=2 q_0'{}^2+2 e^{2 q_0 } 
-\frac 12 \left(q_1'{}^2-q_1'q_2'+q_2'{}^2 +e^{2 q_1-q_2}+e^{2q_2-q_1} \right)+ \epsilon Q^2\,, 
\end{align}
with the constraint
\begin{align}
\label{Hq}
\ti H= 2 q_0'{}^2-2 e^{2 q_0 } 
-\frac 12 \left(q_1'{}^2-q_1'q_2'+q_2'{}^2 -e^{2 q_1-q_2}-e^{2q_2-q_1} \right)-\epsilon Q^2=0\,.
\end{align}

With these variables, the system decomposes into two sectors. 
$q_0$ decouples from other variables and obeys the Liouville equation
\begin{align}
\label{A1Toda}
q_0''=e^{2q_0} \,,
\end{align}
with the conserved energy 
\begin{align}
\label{EL}
E_{\rm L} =q_0'{}^2 - e^{2 q_0} \,. 
\end{align}
The other sector involves $q_1$ and $q_2$, which satisfy the coupled equations
\begin{align}
\label{}
q_1''=e^{2 q_1-q_2 }\,, \qquad q_2 ''= e^{2q_2 -q_1 } \,,  
\end{align}
with the quadratic and cubic integrals 
\begin{align}
\label{ET}
E_{\rm T}=&\, \frac 12 \left(q_1'{}^2-q_1'q_2'+q_2'{}^2 -e^{2 q_1-q_2}-e^{2q_2-q_1}\right)\,, \\
\label{C3}
C_3 =&\, q_1'q_2'(q_1'-q_2')-q_2' e^{2q_1-q_2}+q_1' e^{2q_2-q_1} \,. 
\end{align}
It follows that the Hamiltonian constraint (\ref{Hq}) reduces to 
\begin{align}
\label{constraint}
2E_{\rm L}-E_{\rm T}-\epsilon Q^2=0 \,. 
\end{align}

Introducing 
\begin{align}
\ma C_{ij} \equiv 
\left(
\begin{array}{cc}
2 & -1 \\
-1 & 2  
\end{array}
\right)\,,
\label{CartanA}
\end{align}
equations for $q_i$ ($i=1,2$) can be rewritten as
\begin{align}
\label{A2Toda}
q_i''= \exp \left(\sum_{j=1}^2 \ma C_{ij}q_j \right) \,. 
\end{align}
The matrix $\ma C$ is precisely the Cartan matrix of the semisimple Lie algebra of type $A_2$.
Consequently, the coupled equations for $q_1$ and $q_2$ constitute the one-dimensional $A_2$-Toda system~\cite{Bogoyavlensky}, whose integrability allows the general local solution to be obtained explicitly~\cite{Kostant:1979qu,Olshanetsky,Anderson:1995sz}. 
The two exponential interactions have distinct geometric origins: the term $e^{2q_1-q_2}$ originates from the curvature potential of the Berger sphere, whereas $e^{2q_2-q_1}$ arises from the Routh reduction of the rotational degree of freedom at fixed $j$. The constant shifts in the definitions of $q_1$ and $q_2$ bring these interactions to their canonical normalization.

Since the decoupled Liouville equation for $q_0$ is itself regarded as the rank-one $A_1$-Toda equation, 
the radial equations for $(q_0,q_1,q_2)$ are governed by an $A_1\oplus A_2$ Toda system. In appendix~\ref{app:Lax}, we give a basic review of the Lax pair representation of the $A_1\oplus A_2$ Toda system. 
This direct-sum structure refers to the decoupling of the second-order equations.
Nevertheless, the two sectors remain linked by the Hamiltonian constraint, which relates their conserved energies to the scalar charge through (\ref{constraint}). The distinction between ordinary and phantom scalars therefore enters through this constraint, rather than through a modification of the Toda equations. 
Consequently, the same Toda construction yields exact local
solutions for an ordinary scalar field ($\epsilon=+1$), with
$Q^2=2E_{\rm L}-E_{\rm T}\ge 0$, as well as vacuum solutions
obtained by setting $Q=0$, see appendix~\ref{app:MP}.
The restriction to $\epsilon=-1$ arises only upon imposing the
regular two-ended wormhole boundary conditions below.

To summarize this section, 
once $q_0$ and $q_i$ are determined from the Liouville (\ref{A1Toda}) and Toda equations (\ref{A2Toda}) 
respectively, the spacetime fields are reconstructed as
\begin{align}
\label{metrics}
\D s^2=&\, -\frac{e^{q_2}}{4j^{4/3}}\D t^2+\frac 14 j^{2/3} e^{4q_0-q_1}\D z^2
+\frac 14 j^{2/3} e^{2q_0-q_1} (\sigma_1^2+\sigma_2^2)+\frac 14j^{2/3} e^{q_1-q_2} 
\left(\sigma_3-\omega(z) \D t\right)^2 \,,\\
\label{phis}
\phi=&\, \phi_0 +\epsilon Q z \,, 
\end{align}
where $\phi_0$ is an integration constant. The scalar charge $Q$ is subject to the Hamiltonian constraint (\ref{constraint}). 
Using the Toda equation for $q_2$ (\ref{A2Toda}) in the first equation of (\ref{omegaphip}),  $\omega $ is integrated to give 
\begin{align}
\label{}
 \omega(z)=\omega_0+\frac 1j  q_2'(z)  \,. 
\end{align}
Here,  $\omega_0$ is an integration constant corresponding to
the choice of rotating frame. 
The construction of the spacetime can accordingly be organized into two steps: integrating the $A_1\oplus A_2$ Toda equations locally, and then selecting integration constants that satisfy the Hamiltonian constraint and the desired reality, regularity, and asymptotic conditions.


\subsection{Solving the Toda system for wormholes}
\label{sec:solvingToda}

In this section, we provide a detailed derivation of the solution based on regularity and asymptotic conditions. 
Readers wishing to skip the technical details may proceed directly to section~\ref{sec:wormhole}.

Since the Toda equations are integrable, various methods for solving them have been extensively studied and applied to
black-hole solutions, see e.g., \cite{Gibbons:1987ps,Monni:1995vu,Lu:1996hh,Chemissany:2010ay,Lu:2013uia,Lu:2026lkv,Yun:2026ure}. 
For the Liouville sector, it is convenient to define 
\begin{align}
\label{tau0}
\tau_0(z) \equiv e^{-q_0(z)} \,,  
\end{align}
obeying 
\begin{align}
\label{}
\tau_0'{}^2- \tau_0 \tau_0''=1 \,. 
\end{align}
On account of $(\tau_0''/\tau_0)'=0$, 
the general solution of the Liouville equation (\ref{A1Toda}) is
\begin{align}
\label{q0sol}
\tau_0(z)=\frac{\sinh (\mu_0z+c_0)}{\mu_0 }\,,
\end{align}
where  $\mu_0 $ and $c_0$ are complex constants. 
The Liouville energy (\ref{EL}) is now reduced to $E_{\rm L}=\mu_0^2$, 
and hence $\mu_0$ must be real or purely imaginary. 

Similarly,  the general local solution of the
$A_2$ Toda equations (\ref{A2Toda}) can be expressed as finite sums of exponentials for distinct spectral parameters.
To see this, it is convenient to introduce
$\tau$-functions \cite{Anderson:1995sz}
\begin{align}
\label{tau}
\tau_i(z) \equiv e^{-q_i(z)}\,, 
\end{align}
in terms of which the Toda equations (\ref{A2Toda}) are reduced to bilinear relations 
\begin{align}
\label{taueq}
\tau_1'{}^2-\tau_1\tau_1'' =\tau_2\,, \qquad 
\tau_2'{}^2-\tau_2\tau_2'' =\tau_1\,. 
\end{align}
The solution for the $\tau$-functions can be written as 
\begin{align}
\label{tau1ansatz}
\tau_1 =\sum_{i=1}^3 a_i e^{\lambda_i z}\,,  
\end{align}
where $a_i$ and $\lambda_i$ are constants satisfying 
\begin{align}
\label{}
 \sum_{i=1}^3 \lambda_i =0 \,, \qquad a_1 a_2 a_3 \prod_{i<j} (\lambda_i-\lambda_j)^2 =-1 \,. 
\end{align}
When the spectral parameters $\lambda_i$ are not degenerate, 
the general solutions can be parameterized as \cite{Anderson:1995sz,Lu:2026lkv}
\begin{subequations}
\label{taui}
\begin{align}
\label{q1sol}
\tau_1(z)=&\, \frac{e^{\mu_1 z+c_1}}{(\mu_1-\mu_2)(2\mu_1+\mu_2)}
+\frac{e^{\mu_2 z+c_2}}{(\mu_1-\mu_2)(\mu_1+2\mu_2)}-\frac{e^{-[(\mu_1+\mu_2)z+c_1+c_2]}}{(2\mu_1+\mu_2)(\mu_1+2\mu_2)}
\,, \\
\label{q2sol}
\tau_2(z)=&\, \frac{e^{-(\mu_1 z+c_1)}}{(\mu_1-\mu_2)(2\mu_1+\mu_2)}
+\frac{e^{-(\mu_2 z+c_2)}}{(\mu_1-\mu_2)(\mu_1+2\mu_2)}-\frac{e^{(\mu_1+\mu_2)z+c_1+c_2}}{(2\mu_1+\mu_2)(\mu_1+2\mu_2)}
\,.
\end{align}
\end{subequations}
Here, $\mu_i $ and $c_i$ are complex constants and we have set 
$\lambda_1=\mu_1$, $\lambda_2=\mu_2$ and $\lambda_3=-(\mu_1+\mu_2)$. 
The Hamiltonian constraint (\ref{constraint}) is now reduced to 
\begin{align}
\label{}
4\mu_0^2 -(\mu_1^2+\mu_1\mu_2+\mu_2^2)-2\epsilon Q^2 =0\,. 
\end{align}
The integration constants must be chosen so that the geometry describes
a regular wormhole connecting two asymptotically flat regions. 
In appendix~\ref{app:MP}, we demonstrate the derivation of the Myers-Perry black hole
with equal angular momenta from the Toda system. 

We first focus on the behavior of $q_0$. At both asymptotic ends, 
the metric components (\ref{metric0}) behave as 
\begin{align}
\label{asybehavior}
e^{2A} \to {\rm const.}\,, \qquad e^{2B}\simeq  e^{2C}\to + \infty\,, 
\end{align}
giving $q_0 = A+B+C+\ln 2 \to +\infty$. Since $q_0\to +\infty$  at both asymptotic ends,
it must admit a minimum $q_0^{(\rm min)}$ in the interior, where 
$E_{\rm L}=\mu_0^2= q_0'{}^2 - e^{2 q_0}=- e^{2q_0^{\rm (min)}}<0$. 
Thus, the integration constant $\mu_0$ must be pure imaginary. 
The reality of $q_0$ requires that $c_0$ is also pure imaginary.
Using the shift degree of freedom $z\to z+{\rm const}$, 
one can set $c_0=\pi i/2 $. Writing 
$\mu_0 =2m i$ with $m>0$, 
we obtain 
\begin{align}
\label{}
\tau_0(z)=\frac{\cos(2m z)}{2m} \,, \qquad 
-z_0<z<z_0 \,, \qquad z_0 \equiv \frac{\pi}{4m} \,. 
\end{align}
Here, $z=\pm z_0$ correspond to the two asymptotic ends.

We next translate the asymptotic and regularity requirements into
conditions on the Toda $\tau$-functions (\ref{tau}), in terms of which the metric is rephrased as 
\begin{align}
\label{metrictau}
\D s^2=&\, -\frac{1}{4j^{4/3}\tau_2(z)}\D t^2+\frac{j^{2/3} \tau_1(z)}{4\tau_0(z)^4}\D z^2
+\frac{j^{2/3}\tau_1(z)}{4\tau_0(z)^2} (\sigma_1^2+\sigma_2^2)+\frac { j^{2/3} \tau_2(z)}{4\tau_1(z)} 
\left(\sigma_3-\omega(z) \D t\right)^2 \,,
\end{align}
where $\omega(z)=\omega_0-\tau_2'(z)/[j\tau_2(z)]$. 
Since $\tau_0(z)$ is finite and positive for $-z_0<z<z_0$, the
functions $\tau_i$ must remain positive throughout this interval to 
retain a regular interior with nondegenerate spacelike three-sphere orbits.
At each asymptotically flat end, 
it follows from (\ref{asybehavior}) and (\ref{metrictau}) that
\begin{align}
\label{eq:TauEndpointLimits}
\tau_1(\pm z_0)=0\,,
\qquad
0<\tau_2(\pm z_0)=\tau_1'(\pm z_0)^2<\infty\,,
\end{align}
where the second equation follows from the bilinear identities (\ref{taueq}). 
Thus, the zeros of $\tau_1$ at the two endpoints must be simple.
Positivity in the interior further demands 
\begin{align}
\label{tau1pend}
\tau_1'(-z_0)>0 \,, \qquad \tau_1'(z_0)<0 \,. 
\end{align}

These conditions cannot be satisfied when $\mu_1$ and $\mu_2$ are both real.
To see this, order the spectral parameters in (\ref{tau1ansatz}) as
$\lambda_1<\lambda_2<\lambda_3$ and consider
\begin{align}
e^{-\lambda_2 z}\tau_1(z)
= a_1 e^{(\lambda_1-\lambda_2)z}+a_2+a_3 e^{(\lambda_3-\lambda_2)z}\,.
\end{align}
Its derivative is a sum of two exponentials whose exponents have
opposite signs, and therefore has at most one zero.
Since $e^{-\lambda_2 z}\tau_1$ vanishes at both endpoints and is
positive in the interior, its unique stationary point must be a
maximum, which requires $a_1<0$ and $a_3<0$. 
However, positivity in the interior then implies $a_2>0$, contradicting the condition $a_1a_2a_3<0$.
Thus, the spectral parameters cannot be three distinct real numbers.\footnote{One can similarly verify that the cases with degenerate spectral parameters are incompatible with the boundary conditions.}
Since the quadratic (\ref{ET}) and cubic (\ref{C3}) conserved quantities are real, the three spectral values are roots of a cubic polynomial with real coefficients. See the Lax representation discussed in appendix~\ref{app:Lax}.
Having excluded a real spectrum, we conclude that it must consist of one real root and a complex-conjugate pair.
We therefore write
\begin{align}
\mu_1=a+ib\,,
\qquad
\mu_2=a-ib\,,
\qquad
a\in\mathbb R\,,
\quad b>0\,.
\end{align}
In the parametrization above, reality of the $\tau$-functions is
ensured by
\begin{align}
c_2=\overline{c_1}+i\pi \pmod{2\pi i}\,.
\end{align}
Parametrizing $c_1=d_1+i(d_2+\pi)$ and $c_2=d_1-i d_2$ with $d_1, d_2 \in \mathbb R$, 
the $\tau$-functions then take the manifestly real form
\begin{subequations}
\label{tausol}
\begin{align}
\tau_1(z)
&=\frac{e^{d_1}}{9a^2+b^2}
\left[e^{-2az-3d_1}+e^{az}\left\{\cos(bz+d_2)-\frac{3a}{b}\sin(bz+d_2)\right\}\right]\,,
\\
\tau_2(z)
&=
\frac{e^{-d_1}}{9a^2+b^2}\left[e^{2az+3d_1}+e^{-az}\left\{\cos(bz+d_2)+\frac{3a}{b}\sin(bz+d_2)\right\}\right]\,. 
\end{align}
\end{subequations}

We first derive necessary conditions on the parameters $a,b, d_1, d_2$ for a regular solution.
The explicit expressions for the $\tau$-functions above give
\begin{align}
\label{taupeq}
\left(e^{2az}\tau_1\right)'
&=-\frac{e^{3az+d_1}}{b}\sin(bz+d_2)\,,
\qquad 
\left(e^{-2az}\tau_2\right)'
=-\frac{e^{-3az-d_1}}{b}\sin(bz+d_2)\,.
\end{align}
For a regular solution, both $\tau_1$ and $\tau_2$ must be positive in the interior, 
so that the first bilinear identity (\ref{taueq}) implies
\begin{align}
\label{taulogconcavity}
\left[\ln\!\left(e^{2az}\tau_1\right)\right]''
=(\ln\tau_1)''
=-\frac{\tau_2}{\tau_1^2}<0\,.
\end{align}
Hence $\ln(e^{2az}\tau_1)$ is strictly concave. Since $\tau_1$ has simple zeros at both endpoints,
$\left[\ln\!\left(e^{2az}\tau_1\right)\right]'$ is strictly decreasing, 
and therefore vanishes exactly once.
From (\ref{taupeq}), one finds that successive stationary points of $e^{2az}\tau_1$ are separated by $\pi/b$.
Since the endpoint zeros of $\tau_1$ are simple, the interval must have length strictly less than $2\pi/b$, thus $0<bz_0<\pi$. 

The endpoint conditions $\tau_1(\pm z_0)=0$ can be written as
\begin{align}
\label{tauendpointrelations}
e^{3d_1}\cos\ti b
\left(\cos d_2-\frac{\ti a}{\ti b}\sin d_2\right)
&=-\cosh\ti a\,,
\qquad 
e^{3d_1}\sin\ti b
\left(\sin d_2+\frac{\ti a}{\ti b}\cos d_2\right)
=-\sinh\ti a\,,
\end{align}
where we have denoted $\ti a =3 z_0 a$ and $\ti b=b z_0$. 
Since the first relation implies $\cos\ti b\ne 0$, we have 
\begin{align}
\label{tau1atcenter}
\tau_1(0)
=\frac{z_0^2 e^{-2d_1}}{\ti a^2+\ti b^2}
\left(1-\frac{\cosh\ti a}{\cos\ti b}\right)\,. 
\end{align}
The positivity of $\tau_1$ requires $\cos\ti b<0$.
Combining this with $0<\ti b<\pi$, we obtain
\begin{align}
\label{brange}
\frac{\pi}{2}<\ti b<\pi\,.
\end{align}
Within this range, the endpoint conditions uniquely determine $d_1$ and
determine $d_2$ modulo $2\pi$. 
Taking into account the condition $\sin(2\ti b)<0$, equation (\ref{tauendpointrelations}) gives 
\begin{align}
\label{d1value}
e^{3d_1}
=-\frac{\ti b}{\sin(2\ti b)}
\sqrt{
\frac{2[\cosh(2\ti a)-\cos(2\ti b)]}
{\ti a^2+\ti b^2}
}\,, 
\end{align}
and 
\begin{align}
\label{d2value}
\cos d_2 =
\frac{\sqrt{2}(\tilde b \sin\ti b\cosh\ti a+\tilde a\cos \tilde b\sinh \tilde a)}
{\sqrt{(\tilde a^2+\tilde b^2)
[\cosh (2\ti a)-\cos(2\ti b)]}}\,,
\qquad 
\sin d_2 =
\frac{\sqrt2 (-\tilde a\sin \tilde b \cosh \ti a+\tilde b\cos \tilde b\,\sinh\tilde a)}
{\sqrt{(\tilde a^2+\tilde b^2)
[\cosh (2\ti a)-\cos(2\ti b)]}}\,. 
\end{align}
We choose the representative $-\pi<d_2\le\pi$.

It remains to impose the endpoint slope conditions (\ref{tau1pend}). 
Evaluating (\ref{taupeq}) at the endpoints gives
\begin{align}
\label{tau1p}
\tau_1'(\pm z_0)
=-\frac{e^{d_1\pm az_0}}{b} \sin(d_2\pm\ti b)\,.
\end{align}
For the range (\ref{brange}) and the phase representative (\ref{d2value}), the two slope conditions are equivalent to
$|d_2|<\pi-\ti b$. 
Eliminating $d_2$ by means of (\ref{d2value}), these conditions become
\begin{align}
\label{atbtcond}
\ti b e^{2\ti a}
&>\ti a\sin(2\ti b)+\ti b\cos(2\ti b)\,,
\qquad 
\ti b e^{-2\ti a} >\ti b\cos(2\ti b)-\ti a\sin(2\ti b)\,.
\end{align}
Since $\sin(2\ti b)<0$, the first inequality is automatic for $\ti a\geq 0$, whereas the second is automatic for $\ti a\leq 0$.

Thus far, we have determined the necessary conditions for the regularity at the boundaries. 
These conditions also suffice to ensure positivity throughout the interior.
Indeed, equations (\ref{brange}) and  $|d_2|<\pi-\ti b$ imply
\begin{align}
\label{phaseinterval}
-\pi<d_2-\ti b<0<d_2+\ti b<\pi\,.
\end{align}
It follows from (\ref{taupeq}) that both
$e^{2az}\tau_1$ and $e^{-2az}\tau_2$ increase up to
$z=-d_2/b$ and decrease thereafter.
The former vanishes at both endpoints and is therefore strictly positive in the interior.
The latter is positive at both endpoints, since the bilinear identity gives
$\tau_2(\pm z_0)=[\tau_1'(\pm z_0)]^2>0$, and hence remains strictly positive throughout the interval.
Thus, we establish that conditions (\ref{brange}) and (\ref{atbtcond}) are sufficient to guarantee
\begin{align}
\label{tau12pos}
\tau_1(z)>0 \,, \qquad \tau_2(z)>0 \qquad (-z_0<z<z_0)\,.
\end{align}

Having exhausted the regularity conditions for the Liouville and Toda variables, 
we now determine the remaining constants. 
The Hamiltonian constraint (\ref{constraint}) now gives the scalar charge
\begin{align}
\label{}
\epsilon Q^2=-\frac{8m^2}{3\pi^2}(\ti a^2-3\ti b^2+3\pi^2)\,.
\end{align}
Due to $\ti b<\pi$, this condition demands the phantom field $\epsilon=-1$. 
Requiring the metric to be asymptotically nonrotating at one of the two asymptotic ends $z=z_0$, 
we have $\omega (z_0)=0$, leading to
\begin{align}
\label{}
\omega_0 =\frac{\tau_2'(z_0)}{j \tau_2(z_0)} \,. 
\end{align}
Finally, normalizing the time coordinate so that $e^{2A}\to 1$ at one asymptotic end $z=z_0$, 
we obtain  
\begin{align}
\label{jnorm}
j =\frac{1}{[4\tau_2(z_0)]^{3/4}} \,.  
\end{align}

It follows that the spacetime is characterized by three parameters 
$m$, $\ti a$ and $\ti b$, which collectively determine the ADM mass, equal angular momentum and the 
scalar charge.

\subsection{Wormhole solution}
\label{sec:wormhole}

To make contact with the static solution presented in section \ref{sec:static}, 
we introduce a  new coordinate $u$ by  
\begin{align}
\label{}
u=\frac \pi 2 -2 m z \,, \qquad 0 <u < \pi \,, 
\end{align}
where $u=0$ and $u=\pi $ correspond to the two asymptotic ends. 
We relabel the parameters as 
\begin{align}
\label{}
\beta \equiv -\frac{2\ti a}{3\pi}\,, \qquad 
\nu \equiv \frac{2\ti b}{\pi}\,,
\end{align}
in terms of which the regularity conditions (\ref{brange}) and (\ref{atbtcond}) are equivalently expressed as
\begin{align}
\label{nubetacond}
1<\nu <2 \,, \qquad 
D_\pm \equiv e^{\pm 3\pi \beta}-\cos (\pi \nu)\mp \frac{3\beta}{\nu}\sin(\pi \nu)>0 \,.
\end{align}
The corresponding parameter domain is depicted in figure~\ref{fig:parameter}. 
  
\begin{figure}
\begin{center}
\includegraphics[width=7.5cm]{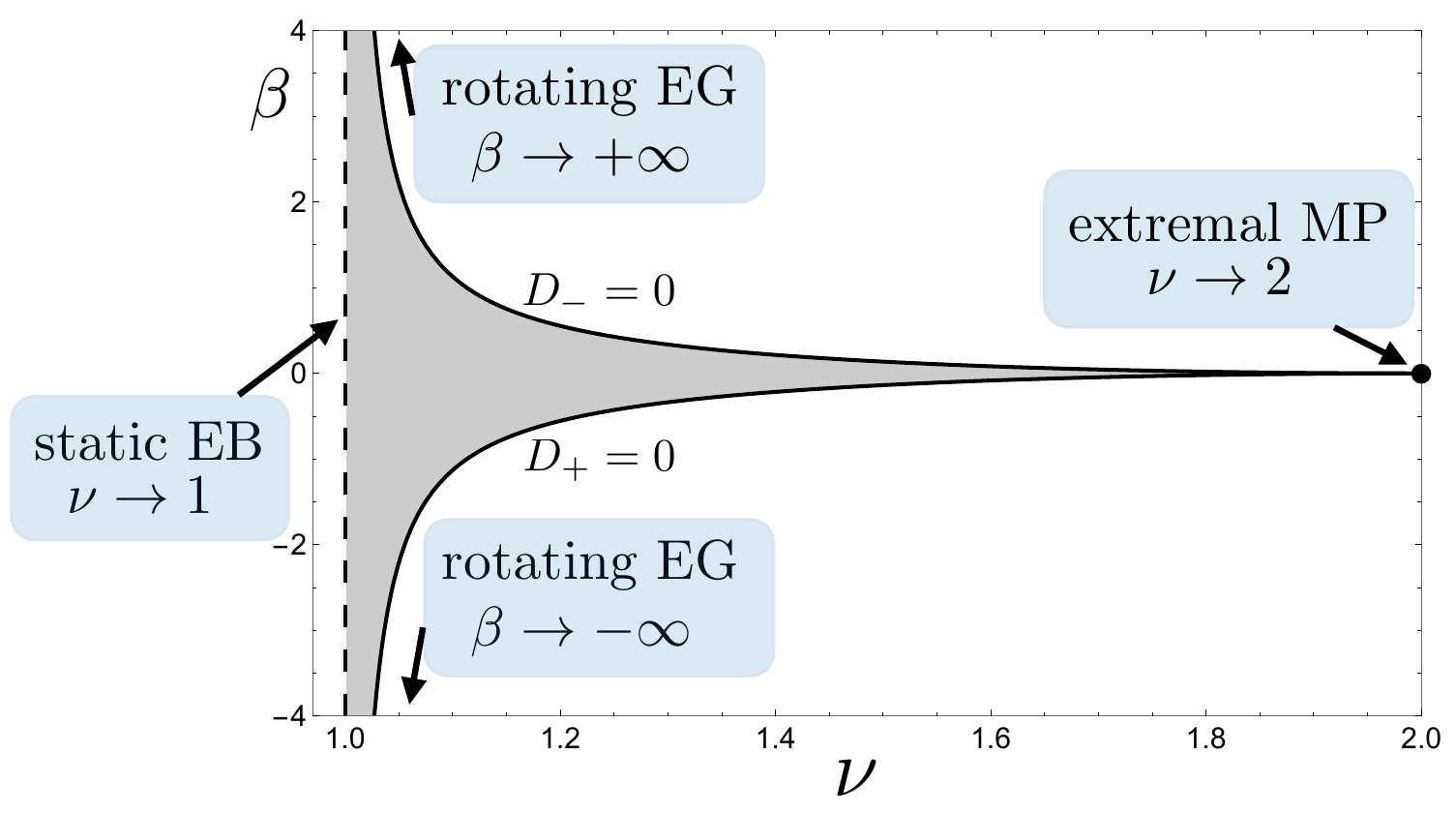}
\caption{Allowed parameter region for the regular rotating wormhole solutions.
The shaded region satisfies $1<\nu<2$ and $D_\pm>0$.
The limits $\nu\to 1$, $\beta\to \pm \infty$ with $\nu\to 1$, and $\nu \to 2$ correspond to 
the static Ellis-Bronnikov solution, the rotating Ellis-Gibbons solution (appendix~\ref{app:EllisGibbons}) and the 
extremal Myers-Perry solution (section~\ref{sec:MP}), respectively.   
}
\label{fig:parameter}
\end{center}
\end{figure}

After reversing the orientation of the rotational Killing direction $\sigma_3\to-\sigma_3$,
the solution is then simplified to 
\begin{align}
\label{metric}
\D s^2=&\, -\frac{e^{-2\beta u}}{H_2(u)}\D t^2+\frac{me^{\beta u}}{4 \sin u}\Bigg\{
H_1(u)\left(\frac{ \D u^2}{\sin^2 u }+\sigma_1^2+\sigma_2^2\right)
+\frac{H_2(u)}{H_1(u)}\left(\sigma_3-\omega(u) \D t \right)^2
\Bigg\}\,, \\
\phi=& \, \pm \sqrt{-\epsilon \frac{4+3\beta^2-\nu^2}{2}}u \,,\label{scalar}
\end{align}
where 
\begin{align}
\label{}
H_1(u)\equiv &\, \frac 1{\sin u}\left\{\frac{1+3 \beta k}{\nu}\sin (\nu u)+k 
\left[e^{-3\beta u}-\cos (\nu u)\right]\right\}\,, \\
H_2(u) \equiv &\, 1+k \left\{[6\beta +(\nu^2+9\beta^2)k]W(u)+\frac{\nu^2+9\beta^2}{\nu}e^{-3\beta u}\sin (\nu u)\right\}\,, 
\\
\omega (u) \equiv &\, \frac{2j}{m^2(\nu^2+9\beta^2)}  \frac{W(u)}{H_2(u)}\,,\\
W(u)\equiv &\, 1- e^{-3\beta u} \left(\cos (\nu u)+\frac{3\beta }{\nu}\sin (\nu u)\right)\,.
\end{align}
Here we have introduced
\begin{align}
\label{}
k \equiv -\frac{\sin (\pi \nu)}{\nu D_- }>0 \,.
\end{align}
The normalization (\ref{jnorm}) now reads
\begin{align}
\label{}
j = m (\nu^2+9\beta^2)\sqrt{km[(1+3\beta k)^2+\nu^2 k^2]}\,. 
\end{align}
With these ingredients at hand, one can verify that the 
metric (\ref{metric}) and the scalar field (\ref{scalar})
solve Einstein's equations and the scalar field equation (\ref{Eineq}).

Since the Liouville and Toda variables are reduced to 
\begin{align}
\label{}
\tau_0(z)=\frac{\sin u}{2m}\,, \qquad 
\frac{\tau_1(z)}{-\tau_1'(z_0)}= \frac 1{2m} e^{\beta u} \sin u  H_1(u) \,, \qquad 
\frac{\tau_2(z)}{\tau_2(z_0)}=e^{2\beta u} H_2 (u) \,, 
\end{align}
the condition (\ref{tau12pos}) amounts to the positivity  of metric functions 
\begin{align}
\label{H12pos}
H_1(u) >0 \,, \qquad H_2(u) >0 \qquad (0<u< \pi) \,,
\end{align}
which is assured by (\ref{nubetacond}). 
Specifically, the boundary values of metric functions are all nonnegative
\begin{align}
\label{metricfunbound}
H_1(0)=H_2(0)=1 \,, \qquad 
H_1(\pi)=\sqrt{H_2(\pi)}=\frac{e^{-3\pi \beta }D_+}{D_-} \,, \qquad W(0)=0 \,, \qquad 
W(\pi)=e^{-3\pi \beta}D_+  \,.  
\end{align}


In the  $\nu\to 1$ limit, we have 
\begin{align}
\label{}
k\to 0 \,, \qquad H_1\to 1 \,,  \qquad H_2\to 1 \,,\qquad \omega \to 0 \,.
\end{align}
Therefore, the solution reduces to the Ellis-Bronnikov wormhole (\ref{EB}) and the 
present solution is regarded as its rotating generalization.


\section{Physical properties}
\label{sec:phys}

In this section, we explore various physical properties of the rotating spacetime (\ref{metric}). 
We demonstrate that the present metric (\ref{metric}) indeed describes a regular rotating wormhole, 
which is asymptotically flat at both ends. The spacetime is horizonless, geodesically complete
and free of curvature singularities. A key ingredient in establishing these properties is the strict positivity (\ref{H12pos}) of $H_1$ and $H_2$ throughout $0<u<\pi$.

\subsection{Absence of Killing horizons and closed causal curves}

A first noteworthy property of the spacetime (\ref{metric}) is the absence of Killing horizons. 
Suppose, to the contrary, that a Killing horizon exists at $u=u_{\rm h}$ and is generated by
$\chi=\partial/\partial t+\Omega_{\rm H}\partial/\partial \psi$, where
$\Omega_{\rm H}$ is a constant.
Since the axial Killing field $\partial/\partial \psi$ is tangent to such a
horizon, normality of $\chi$ would require (see, e.g., section 6.2 of \cite{Heusler:1996jaf})
\begin{align}
\label{horizonnormality}
0=g(\chi,\partial/\partial \psi)
=g_{\psi\psi}\bigl[\Omega_{\rm H}-\omega(u_{\rm h})\bigr]\,.
\end{align}
The strict positivity of $g_{\psi\psi}$ therefore requires
$\Omega_{\rm H}=\omega(u_{\rm h})$ on the putative horizon.
However, this gives
\begin{align}
\label{horizonnorm}
g(\chi,\chi)
=-e^{2A(u_{\rm h})}
+g_{\psi\psi}\bigl[\Omega_{\rm H}-\omega(u_{\rm h})\bigr]^2
=-\frac{e^{-2\beta u_{\rm h}}}{H_2(u_{\rm h})}<0\,,
\end{align}
contradicting the requirement that the horizon generator be null.
Hence, no such Killing horizon exists in the regular interior. 
Equivalently, the Gram matrix of the stationary and axial Killing fields satisfies
\begin{align}
g_{tt}g_{\psi\psi}-g_{t\psi}^{2}
=-e^{2A}g_{\psi\psi}
=-\frac{m e^{-\beta u}}{4\sin u\,H_1(u)}<0\,.
\end{align}
A Killing horizon generated by $\chi$ would require this Gram
matrix to be degenerate, since its null normal must be orthogonal to both Killing fields.
The determinant is invariant under a rigid change of rotating 
coordinates, $\psi\to\psi+\lambda t$.
Thus, the absence of such a Killing horizon is independent of
the choice of rotating frame.

The fact that the squared lapse associated with the constant-$t$ foliation remains
strictly positive can be rephrased as 
\begin{align}
\label{}
g^{\mu\nu}\nabla_\mu t\nabla_\nu t
=g^{tt}=-e^{2\beta u}H_2(u)<0\,. 
\end{align}
Thus, the gradient of $t$ is everywhere timelike.
Since $t\in\mathbb R$ is globally defined, 
it provides a global time function with timelike gradient, establishing that the present spacetime is 
stably causal (see Theorem 3.56 of \cite{Minguzzi:2006sa}).
In particular, $t$ varies strictly monotonically along every future-directed causal curve, which excludes closed causal curves. Consistently, the metric induced on each constant-$t$ hypersurface
is positive definite.

As shown in section~\ref{sec:geodesics}, the present spacetime is timelike
geodesically complete and hence admits no proper extension as a smooth
Lorentzian manifold (see Theorem~1.2 of \cite{Graf:2017xym}).
Thus, the above causality statement applies to the inextendible
spacetime, rather than merely to a proper coordinate patch of a larger
smooth extension.

\subsection{Curvature singularities}

Scalar polynomial curvature singularities can first be diagnosed by  the blow up of curvature invariants.
For instance,  the scalar curvature reads
\begin{align}
\label{}
R=\epsilon (\nabla \phi)^2 = -\frac{2(4+3\beta^2-\nu^2)e^{-\beta u}\sin^3 u}{m H_1(u) } \,. 
\end{align}
Since $H_1$ remains strictly positive and finite, $R$ is regular
throughout the interior and vanishes at both asymptotic ends. 
Although the Kretschmann scalar in the coordinates (\ref{metric})
is too lengthy to be illuminating, it takes a simpler form in terms of the $\tau$-functions. 
As demonstrated in appendix~\ref{sec:curvature}, the Kretschmann scalar
is also finite. 
Since other curvature invariants also do not display any divergent behaviors, 
we conclude that the spacetime is free of curvature polynomial singularities.

\subsection{Asymptotic structure}

To see the asymptotic structure of the spacetime, 
it is convenient to change the coordinate $u$ to $x=m \tan[(\pi-2u)/2]$ by (\ref{U}).  
For $x\to \infty$, one obtains 
\begin{align}
\label{}
-g_{tt}\simeq 1-\frac{m[2\beta+k(9\beta^2+\nu^2)]}{x} \,, \qquad 
g_{t\psi} \simeq -\frac{j}{4x} \,, 
\end{align}
with $g_{xx}\simeq 1/(4x)$, $g_{\vartheta\vartheta}\simeq g_{\psi\psi}\simeq x/4$. 
Thus, the spacetime is asymptotically flat around $x\to +\infty$ with 
the ADM mass and the angular momentum 
\begin{align}
\label{}
M_+=\frac{3\pi}{8G}m[2\beta+k(9\beta^2+\nu^2)]\,, \qquad 
J_{+}^{(\psi)}=\frac{\pi j}{8G}\,.
\end{align}

We next investigate the other asymptotic end $x\to -\infty$. 
For $x\to -\infty$, it is crucial to notice that the angular velocity of the spacetime 
approaches  a nonvanishing constant as~\cite{Dzhunushaliev:2013jja}
\begin{align}
\label{}
\omega =\omega_{-\infty}+O(1/|x|^2)\,, \qquad 
\omega_{-\infty}\equiv \frac{2j  e^{3\pi \beta}D_-^2}{m^2(9\beta^2+\nu^2)D_+}\,. 
\end{align}
Thus, we need to work in a non-rotating frame $\psi=\ti \psi+\omega_{-\infty} t $
to properly assess the conserved quantities.  
Further rescaling the coordinates 
\begin{align}
\label{}
\ti x= -\frac{e^{-2\pi \beta}D_+}{D_-} x \,, \qquad \ti t = \frac{e^{2\pi \beta}D_-}{D_+}t \,, 
\end{align}
the asymptotic expansion $\ti x\to +\infty $ gives 
\begin{align}
\label{}
-g_{\ti t \ti t} \simeq 1+\frac{m e^{-2\pi \beta}}{D_- \ti x}\left[2\beta D_+-kD_- (9\beta^2+\nu^2)\right]\,, \qquad  g_{\ti t \ti \psi} \simeq 
\frac{j}{4\ti x}\,, 
\end{align}
with $g_{\ti x\ti x}\simeq 1/(4\ti x)$, $g_{\vartheta\vartheta}\simeq g_{\ti \psi\ti \psi}\simeq \ti x/4$. 
Thus, 
\begin{align}
\label{}
M_-=-\frac{3\pi m e^{-2\pi \beta}}{8D_- G}\left[2\beta D_+-kD_- (9\beta^2+\nu^2)\right] \,, \qquad 
J_-^{(\ti \psi)}=-\frac{\pi j}{8G}= 
- J_+ ^{(\psi)}\,. 
\end{align}
It should be observed that, unlike in the static case $\nu=1$,
$M_+$ and $M_-$ need not have opposite signs.

The final expression $J_-^{(\ti \psi)}=- J_+ ^{(\psi)}$ is not accidental, but a consequence of the conservation of Komar integrals. 
Since the present spacetime possesses exact Killing symmetries rather than mere asymptotic symmetries, 
the mass and the angular momentum can alternatively be expressed by the integrals
\begin{align}
\label{}
M=-\frac{3}{32\pi G }\int _{S^{3}} \star \D \xi _{(t)} \,, \qquad 
J^{(\psi)}= \frac 1{16\pi G} \int _{S^{3}} \star \D \xi_{(\psi)} \,,
\end{align}
where the same symbols $\xi _{(t)}=\partial/\partial t$ and  $\xi_{(\psi)} =\partial/\partial \psi$ are used for the corresponding Killing one-forms. For the stationary and axial Killing fields considered here,
$\xi^\mu\nabla_\mu\phi=0$. Hence, 
using the Einstein equations (\ref{Eineq}), we have $\nabla_\nu \nabla^\mu \xi^ \nu =R^\mu{}_\nu \xi^\nu =0$. 
Integrating this over the spatial hypersurface extending from $x=-\infty $
to $x=+\infty $, we have 
\begin{align}
\label{Komarintpm}
\int _{S^3_{+\infty}}\star \D \xi + \int _{S^3_{-\infty}}\star \D \xi =0 \,,
\end{align}
where $S^3_{\pm\infty}$ are oriented as boundary components of a constant-$t$ hypersurface, with their normals directed toward their respective asymptotic ends. 
The same globally defined Killing field $\xi$ is used in both integrals.
By taking $\xi=\partial/\partial \psi$,  the Komar angular momentum formula gives  $J_-^{(\ti \psi)}=- J_+ ^{(\psi)}$. 

The corresponding relation between the masses has the same origin but is
more involved, since the asymptotically non-rotating Killing vector at the
second end is
\begin{align}
\label{}
\frac{\partial }{\partial\ti t }=\frac{e^{-2\pi \beta}D_+}{D_-}\left(
\frac{\partial }{\partial t}+\omega_{-\infty}\frac{\partial}{\partial\psi}\right) \,.
\end{align} 
Substituting this into (\ref{Komarintpm}), one finds a Smarr-like relation 
\begin{align}
\label{}
M_-=  \frac{e^{-2\pi \beta}D_+}{D_-}\left(-M_++\frac 32 \omega_{-\infty} J_+^{(\psi)} \right)\,. 
\end{align}
Since the masses $M_+$ and $M_-$ are associated with different time-translation Killing fields, each canonically normalized and nonrotating at its own asymptotic end,  they therefore need not coincide.
The relation above is precisely the consequence of Stokes' theorem for a fixed global Killing field, expressed in terms of these two asymptotic frames.

\subsection{Throat}
\label{sec:throat}

We next turn our attention  to the throat, which bridges the two asymptotic ends. 
The area of a $\{t={\rm const}, u={\rm const.}\}$ three-dimensional group orbit is given by
\begin{align}
\label{area}
\ma A(u)=2\pi^2 m^{3/2} e^{\frac32\beta u}
\sqrt{\frac{H_1(u)H_2(u)}{\sin^3u}}\,.
\end{align}
Within the present cohomogeneity-one foliation, 
we define the throat as a minimum of $\ma A(u)$.
Accordingly, its location $u=u_{\rm th}$ satisfies
\begin{align}
\label{throatcondition}
\left.3\beta-3\cot u+\frac{\D}{\D u}\ln\bigl(H_1H_2\bigr)\right|_{u=u_{\rm th}}
=0\,.
\end{align}
Although this equation cannot in general be solved explicitly, the Toda
structure allows us to prove that the throat exists uniquely.
To this end, we define
\begin{align}
\label{P1P2}
P_1(u) \equiv\ln\left(\frac{2H_1(u)^2}{H_2(u)}\right)\,,\qquad 
P_2(u)\equiv \ln\left(\frac{2m^3e^{3\beta u}H_2(u)^2} {j^2\sin^3u\,H_1(u)}
\right)\,,
\end{align}
satisfying 
\begin{align}
\label{P1P2eq}
\frac{\D^2P_1}{\D u^2}=
\frac{2}{\sin^2u}\left(1-2e^{-P_1}+e^{-P_2} \right)\,,
\qquad 
\frac{\D^2P_2}{\D u^2}
&=\frac{2}{\sin^2u}\left(1+e^{-P_1}-2e^{-P_2}\right)\,.
\end{align}
These relations are directly inherited from the Toda $\tau$-functions and their 
bilinear relations (\ref{taueq}).
In terms of these variables, the area (\ref{area}) can be expressed compactly as
\begin{align}
\label{AP}
\ma A(u) =\pi^2 j e^{[P_1(u)+P_2(u)]/2}\,,
\end{align}
and hence
\begin{align}
\label{A2uu}
\frac{\D^2}{\D u^2}\ln \ma  A(u)=\frac{1}{\sin^2u}\left(2-e^{-P_1}-e^{-P_2}\right)\,.
\end{align}
To demonstrate the uniqueness of the throat, 
we now prove the following lemma. 

\begin{lemma}
The functions $P_1$ and $P_2$ defined above satisfy
\begin{align}
\label{Ppositive}
P_1(u)>0\,,\qquad P_2(u)>0 \qquad (0<u<\pi)\,.
\end{align}
\end{lemma}
\noindent{\it Proof.}
The endpoint behavior inherited from the asymptotically flat geometry is
\begin{align}
\label{Pendpoints}
P_1\to \ln 2\,,
\qquad P_2\to +\infty \qquad (u\to0,\pi)\,.
\end{align}
Suppose, to the contrary, that either $P_1$ or $P_2$ becomes nonpositive somewhere in the interior.
Define
\begin{align}
p_i\equiv\min_{0<u<\pi}P_i(u)
\qquad (i=1,2)\,.
\end{align}
Although $P_1$ and $P_2$ have different definitions, the system (\ref{P1P2eq}) is invariant under the interchange $P_1\leftrightarrow P_2$. 
Moreover, the following argument uses only this system and the fact that
both $P_1$ and $P_2$ are positive near the two endpoints.
It therefore suffices to consider the case 
\begin{align}
p_1\le p_2\,,\qquad p_1\le0\,.
\end{align}
Since $P_1$ is positive near both endpoints, its minimum is attained at
some interior point $u=u_*$, where
\begin{align}
\label{P1us}
P_1(u_*)=p_1\,,\qquad
\left.\frac{\D P_1}{\D u}\right|_{u=u_*}=0\,,\qquad
\left.\frac{\D^2 P_1}{\D u^2}\right|_{u=u_*}\ge 0\,.
\end{align}
Furthermore,
\begin{align}
P_1(u_*)=p_1\le p_2\le P_2(u_*)\,.
\end{align}

Suppose first that $P_1(u_*)<0$. 
Due to $P_1(u_*)\le P_2(u_*)$, we have $e^{-P_2(u_*)}\le e^{-P_1(u_*)}$. 
Equation (\ref{P1P2eq}) therefore gives
\begin{align}
\left.\frac{\D^2 P_1}{\D u^2}\right|_{u=u_*}
&=
\frac{2}{\sin^2u_*} \left(1-2e^{-P_1(u_*)}+e^{-P_2(u_*)}\right)
\le \frac{2}{\sin^2u_*}  \left( 1-e^{-P_1(u_*)} \right) <0\,,
\end{align}
which contradicts the last condition in (\ref{P1us}).
Thus, $p_1$ cannot be negative.

It remains to exclude the limiting possibility $p_1=0$.
In this case, $P_1(u_*)=0$ and $P_2(u_*)\ge0$, so that
\begin{align}
0 \le \left.\frac{\D^2 P_1}{\D u^2}\right|_{u=u_*}
&=\frac{2}{\sin^2u_*} \left(-1+e^{-P_2(u_*)}\right) \le0\,.
\end{align}
Hence all inequalities must be saturated, which implies $P_2(u_*)=0$. 
Since $p_2\ge p_1=0$, this means that $u_*$ is also a minimum of $P_2$.
Consequently,
\begin{align}
P_1(u_*)=P_2(u_*)=0\,,
\qquad
\left.\frac{\D P_1}{\D u}\right|_{u=u_*}
=\left.\frac{\D P_2}{\D u}\right|_{u=u_*}=0\,.
\end{align}
With these initial data, the uniqueness of the initial-value problem for the coupled second-order equations 
(\ref{P1P2eq}) implies that $P_1=P_2=0$ identically, in contradiction with the endpoint behavior (\ref{Pendpoints}).
Thus, equation (\ref{Ppositive}) follows.
\qed

\bigskip 
Since we have established the positivity of $P_i$, 
it then follows immediately from (\ref{A2uu}) that
\begin{align}
\label{logAconvex}
\frac{\D^2}{\D u^2}\ln \ma A(u)>0 \qquad (0<u<\pi)\,.
\end{align}
Hence, $\ln{\cal A}(u)$ is strictly convex throughout the wormhole
interior.  Moreover, since $H_1$ and $H_2$ remain finite and positive
at both asymptotic ends, we have 
\begin{align}
\frac{\D}{\D u}\ln{\ma A}=-\frac{3}{2u}+O(1)\quad (u\to0+)\,,
\qquad 
\frac{\D}{\D u}\ln{\cal A} =\frac{3}{2(\pi-u)}+O(1) 
\quad (u\to\pi-)\,.
\end{align}
Since $\D\ln{\cal A}/\D u$ is strictly increasing, it vanishes
precisely once in $0<u<\pi$.
Consequently, ${\cal A}(u)$ possesses a unique stationary point,
which is necessarily a strict minimum.  The spacetime therefore possesses a unique throat, whose location
is determined by (\ref{throatcondition}).

Having established the positivity of $P_i$, we obtain a universal upper bound 
\begin{align}
\label{JAthineq}
\frac{G|J ^{(\psi)}|}{{\cal A}_{\rm th}}=\frac 1{8\pi} 
e^{-[P_1(u_{\rm th})+P_2(u_{\rm th})]/2}<\frac 1{8\pi} \,, 
\end{align}
where $J$ refers to $J_+ ^{(\psi)}$ or $J_- ^{(\tilde\psi)}$. 
This means that once the area of the throat is fixed, 
the angular momentum of the spacetime cannot be arbitrarily large. 
We will return to this point in section \ref{sec:MP}.


\subsection{Ergoregion}

Because the two asymptotic regions rotate relative to each other, there is no single stationary Killing vector that is canonically normalized and nonrotating at both ends. Consequently, the notion of an ergoregion is asymptotic-end dependent. 

For definiteness of the argument, we restrict ourselves to 
 the ergoregion defined with respect to the
canonically normalized stationary Killing field $\partial/\partial t$ at the $u=0$ end. 
The ergosurface is then determined by 
\begin{align}
\label{calF}
\ma F(u)\equiv \sin(\pi\nu)W(u)+e^{-3\beta u}D_-\sin(\nu u)=0\,,
\end{align}
where $g_{tt}=-e^{\beta u}\ma F(u)/[\nu D_- \sin u H_1(u)]$ and hence the ergoregion corresponds to $\ma F(u)<0$. 
To analyze the zeros of ${\cal F}$, we first establish the positivity
of $W$ in the interior.
Due to $1<\nu<2$ and 
\begin{align}
\frac{\D W}{\D u}=
\frac{\nu^2+9\beta^2}{\nu}
e^{-3\beta u}\sin(\nu u)\,,
\end{align}
the function $W$ increases for $0<u<\pi/\nu$ and decreases for $\pi/\nu<u<\pi$.
Together with $W(0)=0$ and $W(\pi)=e^{-3\pi\beta}D_+>0$,
this implies $W(u)>0$ throughout $0<u<\pi$.
Consequently, since $\sin(\pi\nu)<0$ and $D_->0$, we have
${\cal F}(u)<0$ for $\pi/\nu\leq u<\pi$.
Any ergosurface must therefore lie in $0<u<\pi/\nu$.

We now show that there is precisely one such surface.
On this interval, we have 
\begin{align}
\frac{e^{3\beta u}}{\sin(\nu u)}
\frac{\D \ma F}{\D u}
=\nu D_-\cot(\nu u)-3\beta D_-+\frac{9\beta^2+\nu^2}{\nu}\sin(\pi\nu)\,.
\end{align}
Due to $D_->0$, 
the right-hand side is strictly decreasing from $+\infty$ to $-\infty$.
The prefactor on the left-hand side is positive.
Consequently, there is precisely one point
$u_{\rm m}\in(0,\pi/\nu)$ at which $\D \ma F/\D u=0$. It follows that 
$\ma F$ is strictly increasing for $0<u<u_{\rm m}$ and strictly decreasing for $u_{\rm m}<u<\pi/\nu$.
Moreover,
\begin{align}
\ma F(0)=0\,,
\qquad
\ma F\left(\pi/\nu\right)
=\sin(\pi\nu) \left(1+e^{-3\pi\beta/\nu}\right)<0\,.
\end{align}
Thus, $\ma F(u_{\rm m})>0$, and its decreasing branch crosses zero
precisely once.
Together with the negativity of $\ma F$ for $\pi/\nu\leq u<\pi$,
this proves that the ergosurface is unique and satisfies
\begin{align}
0<u_{\rm erg}<\frac{\pi}{\nu}\,.
\end{align}
We therefore conclude
\begin{align}
\ma F(u)>0
\quad (0<u<u_{\rm erg})\,,
\qquad
\ma F(u)<0
\quad (u_{\rm erg}<u<\pi)\,.
\end{align}
This shows that the ergoregion associated with $\partial/\partial t$
extends from the ergosurface all the way to the opposite asymptotic end $u=\pi$.

Note also that the zero of $\ma F$ is simple, with
$\left.\D \ma F/\D u\right|_{u=u_{\rm erg}}<0$. 
This means that the ergosurface is not evanescent-type~\cite{Gibbons:2013tqa}.
Evanescent ergosurfaces are known to lead to strong trapping effects and associated instabilities
for various black objects and solitons in five dimensions~\cite{Eperon:2016cdd}.

\subsection{Geodesics}
\label{sec:geodesics}

We consider geodesics  governed by 
$p^\mu\nabla_\mu p^\nu=0$ and $p^\mu p_\mu=-\delta$, 
where $\delta =1$ for timelike geodesics, $\delta =0$ for null geodesics and 
$\delta=-1$ for spacelike geodesics. 
We analyze the geodesic motion by the Hamilton-Jacobi equation
\begin{align}
\label{}
-\frac{\partial \ma S}{\partial \lambda}=\frac 12 g^{\mu\nu}\frac{\partial \ma S}{\partial x^\mu }\frac{\partial \ma S}{\partial x^\nu }\,,
\end{align}
where $\lambda$ is an affine parameter of the geodesics. 
Considering the separable form 
\begin{align}
\label{}
\ma S= \frac 12 \delta \lambda - E t +L_\varphi \varphi +L_\psi \psi +\ma S_u(u)+\ma S_\vartheta (\vartheta ) \,, 
\end{align}
one finds the decoupled equations 
\begin{align}
\label{Sth}
\left(\frac{\D \ma S_\vartheta }{\D \vartheta}\right)^2 =\,& \ma K-\frac{(L_\varphi- L_\psi \cos\vartheta)^2}{\sin^2\vartheta } \,, \\
\label{Su}
\left(\frac{\D \ma S_u }{\D u}\right)^2 =\,& \frac{me^{\beta u}H_1(u)}{4\sin^3 u} \ma R(u) \,, 
\end{align}
where $\ma K(\ge 0)$ is a separation constant and 
\begin{align}
{\cal R}(u)\equiv e^{2\beta u}H_2(u)\bigl[E-L_\psi\omega(u)\bigr]^2-\delta
-\frac{4e^{-\beta u}\sin u}{m}\left[\frac{H_1(u)}{H_2(u)}L_\psi^2+\frac{\ma K}{H_1(u)}\right]\,,
\end{align}
Equation (\ref{Sth}) corresponds to the standard angular equation associated with the ${\rm SU}(2)$-invariant three-sphere. 
The ``radial'' equation takes the form
\begin{align}
\left(\frac{\D u}{\D\lambda}\right)^2
=\frac{4e^{-\beta u}\sin^3u}{mH_1(u)} {\cal R}(u)\,.
\end{align}
Since $H_1$ and $H_2$ are strictly positive and all metric functions
are smooth for $0<u<\pi$, ${\cal R}(u)$ has no singularity in the
wormhole interior.  Its simple zeros correspond to ordinary radial turning points.
Multiple zeros may instead describe constant-$u$ geodesics or limiting
orbits, but do not constitute singularities of the geodesic equations. 

It remains to examine the two asymptotic ends.
Since ${\cal R}(u)$ remains finite there, the radial equation gives
\begin{align}
\left(\frac{\D u}{\D\lambda}\right)^2
&=O(u^3)
\quad (u\to0)\,,
\qquad 
\left(\frac{\D u}{\D\lambda}\right)^2=O\bigl((\pi-u)^3\bigr)
\quad (u\to\pi)\,.
\end{align}
Consequently,
\begin{align}
\left|\frac{\D\lambda}{\D u}\right|
&\gtrsim u^{-3/2}
\quad (u\to0)\,,
\qquad 
\left|\frac{\D\lambda}{\D u}\right|
\gtrsim(\pi-u)^{-3/2}
\quad (u\to\pi)\,,
\end{align}
and the affine parameter required to reach either asymptotic end
diverges.  

To exclude an endpoint at finite affine parameter in the interior, consider
any compact radial interval $[u_1,u_2]\subset(0,\pi)$.
The conserved quantities give
\begin{align}
\frac{\D t}{\D\lambda} =e^{2\beta u}H_2(u)
\bigl[E-\omega(u)L_\psi\bigr]\,,
\end{align}
which is bounded on this interval.
The normalization condition likewise bounds the positive-definite norm
of the spatial projection of the tangent vector.
Thus, in regular coordinate charts on the three-sphere, the tangent vector
remains bounded, and the geodesic equations admit continuation through
any putative interior endpoint at finite affine parameter.
Together with the infinite affine distance to both asymptotic ends,
this establishes geodesic completeness.

\subsection{Reflection symmetric case: $\beta =0$}

In the special case $\beta =0$, the metric is considerably simplified. 
Explicitly, we have 
\begin{align}
\label{}
H_1(u)=\frac{2\sin(\nu u/2)\sin[\nu(\pi-u)/2]}{\nu\sin(\nu\pi/2)\sin u}\,, \qquad 
H_2(u)=\frac{\sin^2(\nu u/2)+\sin^2[(\pi-u)\nu/2]}{\sin^2 (\pi \nu/2)}\,, \qquad 
W(u)=2 \sin^2 (\nu u/2)\,, 
\end{align}
with 
\begin{align}
\label{}
k=-\frac{\cot (\pi \nu/2)}{\nu}\,, \qquad 
j=\frac{m\nu \sqrt{-m\nu \cot(\pi\nu/2)}}{\sin (\pi \nu/2)}\,, \qquad 
D_+=D_- = 2 \sin^2 (\pi\nu/2) \,, \qquad 
\omega_{-\infty} =2\sqrt{\frac{-2 \sin(\pi \nu)}{m\nu}}\,. 
\end{align}
In particular, we have $H_1(\pi-u)=H_1(u)$, $H_2(\pi-u)=H_2(u)$ and 
$\omega(\pi-u)=\omega_{-\infty}-\omega(u)$. 
Thus, $H_1$ and $H_2$ are symmetric about $u=\pi/2$, whereas
$\omega-\omega_{-\infty}/2$ is antisymmetric.
These identities imply that the metric is invariant under
$ (t,u,\vartheta,\varphi,\psi) \mapsto (t,\pi-u,\pi-\vartheta,\varphi,-\psi+\omega_{-\infty}t)$. 
The geometry with $\beta=0$ therefore admits a reflection symmetry interchanging
the two asymptotic ends.
Indeed, the ADM masses at the two ends coincide
\begin{align}
\label{}
M_+=M_-=-\frac{3 \pi m}{8G}\nu \cot (\pi \nu/2)\,. 
\end{align}
Note that $M_+=M_-$, whereas the angular momenta at the two ends have
opposite signs, $J_+^{(\psi)}=-J_-^{(\tilde\psi)}$.
This difference originates from the nonvanishing relative rotation
$\omega_{-\infty}$ between the two asymptotic frames.

In particular, the throat equation (\ref{throatcondition}) is solved to give 
$u_{\rm th}=\pi/ 2$, hence
\begin{align}
\label{}
\ma A_{\rm th}=2\pi^2 m^{3/2} \sqrt{\frac{\tan(\pi \nu/4)}{\nu[1+\cos(\pi \nu /2)]}}\,. 
\end{align}
From (\ref{calF}), the locus of the ergosurface is $u_{\rm erg}=\pi (2-\nu)/\nu$. 
Its position relative to the throat $u_{\rm th}=\pi/2$ is then summarized as in table~\ref{table:ergoth}. 
Thus, the ergosurface crosses the throat at $\nu=4/3$.

\begin{table}[t]
\begin{center}
\begin{tabular}{c|c}
\hline
$\nu$ & ~Location of the ergosurface \\ \hline\hline
$1<\nu<4/3$ & $u_{\rm th}<u_{\rm erg}<\pi$ \\\hline
$\nu=4/3$ & $u_{\rm erg}=u_{\rm th}$ \\\hline
$4/3<\nu<2$ & $0<u_{\rm erg}<u_{\rm th}$ \\ \hline
\end{tabular}
\end{center}
\caption{Range of $\nu$ and the relative location of the throat and the ergosurface in the 
$\beta=0$ case.}
\label{table:ergoth}
\end{table}

Moreover, one can verify the Smarr-like formula
\begin{align}
\label{}
M_+=\frac 32 \omega_{\rm th} J_+^{(\psi)} \,,  
\end{align}
where $\omega_{\rm th}=\omega (\pi/2)=\omega_{-\infty}/2$ denotes the 
angular velocity of the throat \cite{Dzhunushaliev:2013jja}.

\subsection{Extremal Myers-Perry limit: $\nu \to 2$}
\label{sec:MP}

Let us investigate the boundary $\nu\to2$ of the regular parameter region (see figure~\ref{fig:parameter}).  
We show that this limit gives rise to the extremal Myers-Perry solution with 
equal angular momenta.\footnote{In appendix~\ref{app:EllisGibbons}, we discuss another, singular scaling limit that provides a boundary configuration $E_{\rm L}=0$
with a nontrivial scalar field. }  

Setting
\begin{align}
\label{MPlimit1}
\nu=2-\varepsilon\,,\qquad \varepsilon\to 0 +\,,
\end{align}
we find
\begin{align}
D_\pm=\frac{\pi^2}{2}\varepsilon^2 \pm3\pi\beta+O(\varepsilon^3,\beta\varepsilon,\beta^2)\,.
\end{align}
The simultaneous regularity conditions $D_\pm>0$ therefore require
$\beta=O(\varepsilon^2)$.  We accordingly write
\begin{align}
\beta=\frac{\pi}{6}\hat\beta \varepsilon^2\,,
\qquad 
D_\pm=\frac{\pi^2}{2}\left(1\pm\hat\beta\right)\varepsilon^2+O(\varepsilon^3)\,,
\label{MPlimit2}
\end{align}
and hence $-1<\hat\beta<1$. 
To keep the ADM mass and angular momentum finite and nonzero, we further scale
\begin{align}
\label{MPlimit3}
m=\pi\left(1-\hat\beta\right)a^2\varepsilon\,,
\end{align}
where $a$ is a constant. 
With this scaling, the limit of the solution is free from the scaling parameter $\varepsilon$, 
and the scalar field vanishes $\phi\to 0$. 
After setting $x=r^2-a^2$ on the $x>0$ side, 
the metric reduces to the exterior region of an extremal equal-angular-momenta
Myers-Perry solution (\ref{MPmetric}) with $\mu=4a^2$. 
Remarkably, all dependence on $\hat\beta$ drops out. 
The same limiting geometry is obtained from the opposite asymptotic end
after passing to its canonically normalized nonrotating coordinates.

Thus, although $\hat\beta$ parametrizes different directions by which the
regular wormhole family approaches the corner $(\beta,\nu)=(0,2)$ in figure~\ref{fig:parameter},
it does not survive as a parameter of the limiting extremal Myers-Perry geometry.

Our result analytically explains the results of \cite{Uemichi:2026dzb}. 
The asymmetry parameter used in \cite{Uemichi:2026dzb}  is
related to our parameters by
\begin{align}
a_{-\infty}
=2\pi\beta+\ln\frac{D_-}{D_+} \,. 
\end{align}
Hence, in the above scaling limit, $\hat\beta=-\tanh(a_{-\infty}/2)$.
It has been reported in \cite{Uemichi:2026dzb} that the nonextremal 
Myers-Perry solution is never approached from the wormhole solution. 
This can be understood directly from the Toda description. 
For the regular wormhole branch, the Liouville energy is negative,
$E_{\rm L}=-4m^2<0$, 
whereas the nonextremal Myers-Perry branch has
$E_{\rm L}=\mu_0^2=\mu(\mu -4a^2)>0$ (see appendix~\ref{app:MP}). 
A continuous limit from the wormhole family to a vacuum Myers-Perry solution can therefore occur only at the common boundary $E_{\rm L}=0$. On the Myers-Perry branch, this condition is precisely the extremal limit $\mu_0\to0$. 
The scaling limit discussed above realizes this boundary explicitly and yields the exterior region of the extremal equal-angular-momenta Myers-Perry black hole.

In section~\ref{sec:throat}, we derived the universal upper bound (\ref{JAthineq}) on the angular momentum per unit throat area. 
For the Myers-Perry family with equal angular momenta, we have an analogous inequality
\begin{align}
\label{JAHineq}
\frac{G|J^{(\psi)}|}{A_{\rm H}} \le \frac{1}{8\pi}\,,
\end{align}
where equality is attained in the extremal limit. 
It should be emphasized, however, that $A_{\rm H}$ is the area of the event horizon rather than that of a throat. 
Indeed, within the equal-angular-momenta Myers-Perry family, the area (\ref{maAMP}) of the constant-$t$, constant-$r$ group orbits is strictly increasing throughout the exterior region, and $A_{\rm H}$ is its limiting value at the horizon. Thus, $A_{\rm H}$ and $A_{\rm th}$ have different geometric definitions before the scaling limit is taken.

Some care is therefore required when comparing (\ref{JAthineq}) with (\ref{JAHineq}). 
In the scaling limit (\ref{MPlimit1})--(\ref{MPlimit3}), one finds that the area function (\ref{area}) of the wormhole satisfies
\begin{align}
\ma A(u) \to  8\pi^2 a^3\,,
\end{align}
for any fixed $0<u<\pi$. 
In particular, the throat area approaches the area of the extremal Myers-Perry horizon,
$\ma A_{\rm th}\to A_{\rm H}$. 
Crucially, this limiting value is obtained {\it without using the throat condition} (\ref{throatcondition}). 
Indeed, the latter degenerates in this scaling limit and enters only at subleading order in determining the location of the throat. 
Thus, despite the different geometric meanings of $\ma A_{\rm th}$
and $A_{\rm H}$, the bounds (\ref{JAthineq}) and (\ref{JAHineq})
coincide in the extremal Myers-Perry limit.

\section{Conclusion}
\label{sec:conc}

In this paper, we have constructed an exact solution describing a rotating wormhole with equal angular momenta in the five-dimensional Einstein theory coupled to a phantom scalar field. 
The solution connects two asymptotically flat regions and provides an exact rotating generalization of the Ellis-Bronnikov wormhole. 
The availability of a closed-form solution is particularly valuable in this problem, since it allows us to establish analytically a number of properties that have previously been observed through numerical investigations~\cite{Dzhunushaliev:2013jja,Uemichi:2026dzb}. 
In particular, we have shown that the spacetime is horizonless, stably causal, free of curvature singularities, and geodesically complete. 
We have also clarified the asymptotic charges at the two ends and the relative rotation between their asymptotic rest frames, and have established the existence and uniqueness of the ergosurface associated with a stationary frame that is nonrotating at one end.

The central ingredient underlying our construction is the integrable $A_1\oplus A_2$ Toda structure governing the radial field equations. 
Its role, however, goes considerably beyond providing a convenient method of obtaining the exact solution. 
The Toda variables enable us  to implement the regularity and asymptotic-flatness conditions analytically and thereby determine the parameter region in which the local solution describes a globally regular wormhole. 
More importantly, the differential and bilinear identities inherited from the Toda equations provide direct control over the global geometry. 
They allow us to prove the positivity properties of the metric functions, establish the existence and uniqueness of the wormhole throat, and derive a sharp upper bound on the angular momentum for a fixed throat area. 
These results would be difficult to infer  from the final expression of the  metric (\ref{metric}). 
Thus, in the present problem, Toda integrability serves not only as a tool for solving the field equations, but also as an efficient framework for extracting nontrivial geometric information from the resulting spacetime.
As demonstrated in appendix~\ref{app:odd}, 
the analogous Cartan-type Toda structure does not persist in higher odd dimensions
$D\ge 7$, even though the cohomogeneity-one ansatz is preserved.

Several questions remain open.
It would be interesting to understand whether the present solution admits an alternative derivation, for instance via a Kerr-Schild-type formulation~\cite{Nozawa:2025yac}.
A more pressing physical question concerns the stability and quasinormal-mode spectrum \cite{Khoo:2024yeh} of the spacetime.
Since the present solution possesses a ${\rm U}(2)$ spatial symmetry, its perturbation equations may be analyzed by adapting the framework developed in~\cite{Murata:2007gv}.
It would be particularly interesting to provide a semi-analytic check of the numerical stability results reported in~\cite{Dzhunushaliev:2013jja}.
The presence of an ergoregion also raises questions concerning rotational energy extraction.
In particular, it would be worthwhile to investigate the Penrose process~\cite{Nozawa:2005eu}, as well as the possibility of superradiant amplification and associated instabilities.
We leave these issues for future work.

\acknowledgments
We would like to thank Keiya Uemichi for useful comments. 
This work is supported by KAKENHI Grant Number 25K07309.
The author proposed several proof strategies to GPT-6 Pro, which helped testing these approaches, algebraic checks, literature review, and the development of the draft.

\appendix

\section{Toda system}
\label{app:Lax}

\subsection{Lax representation}

We briefly exhibit the Lax representation underlying the integrability used in section~\ref{sec:solvingToda}.
Although this structure is standard for finite-dimensional Toda systems, 
it is useful here because it gives a direct interpretation of the conserved
quantities and of the spectral parameters entering the explicit solution.
For details, we refer the reader to the comprehensive review \cite{Olshanetsky:1981dk}.

The spectral interpretation of a Lax representation follows from
the auxiliary linear system
\begin{align}
\label{eq:lax1}
\mas L \psi 
=\lambda \psi \,,
\qquad
\frac{\D \psi}{\D z }=\mas M \psi \,, 
\end{align}
where $\psi$ is an auxiliary vector, $\lambda$ is a $z$-independent spectral parameter, 
and $\mas L$ and $\mas M$ are matrices. 
Requiring compatibility as a matrix identity gives the Lax equation
\begin{align}
\label{eq:lax}
\frac{\D \mas L}{\D z} =[\mas M, \mas L] \,. 
\end{align}
Conversely, this equation implies
\begin{align}
\mas L(z)
=U(z)\mas{L}(z_*)U(z)^{-1}\,, \qquad 
\frac{\D U}{\D z}=\mas M U \,, 
\end{align}
where  $U(z_*)=\mathbf{1}$.
Thus, the flow is isospectral, and $\operatorname{Tr}(\mas{L}^n)$
is conserved for every positive integer $n$.

For the Liouville sector,  the following $2\times 2$ matrices 
constitute a Lax pair
\begin{align}
\label{}
\mathscr L_{\rm L} =\left(
\begin{array}{cc}
-q_0'      & 1   \\
 -e^{2q_0}     &   q_0'
\end{array}
\right)\,, 
\qquad 
\mathscr M_{\rm L} =\left(
\begin{array}{cc}
-q_0'      & 1   \\
0    &   q_0'
\end{array}
\right)\,.
\end{align}
The Liouville energy is expressed as $E_{\rm L}={\rm Tr}(\mathscr  L^2_{\rm L})/2$, 
which is conserved as a consequence of the isospectral evolution.
The nonlinear Liouville dynamics are thus encoded in the compatibility
of the auxiliary linear system (\ref{eq:lax1}). This interplay between nonlinear
dynamics and linear spectral theory lies at the heart of the
Lax approach to integrability.

Similarly, a Lax pair for the $A_2$ Toda system is given by
\begin{align}
\label{}
\mathscr L_{\rm T} =\left(
\begin{array}{ccc}
  -q_1'    & 1 & 0  \\
  -e^{2q_1-q_2}     &   q_1'-q_2' & 1 \\
  0 & - e^{2q_2-q_1} & q_2' 
\end{array}
\right)\,, \qquad 
\mathscr M_{\rm T} =\left(
\begin{array}{ccc}
  -q_1'    & 1 & 0  \\
0   &   q_1'-q_2' & 1 \\
  0 & 0 & q_2' 
\end{array}
\right)\,.
\end{align}
The quadratic and cubic invariants are given by 
\begin{align}
\label{}
E_{\rm T} =\frac 14 {\rm Tr}(\mathscr  L_{\rm T} ^2) \,, \qquad 
C_3=-\frac 13 {\rm Tr}(\mathscr  L_{\rm T} ^3) \,. 
\end{align}
The conserved spectral data of the Toda sector are encoded in
the characteristic polynomial
\begin{align}
\label{TodaPoly}
\det(\lambda {\bf 1}-\mathscr  L_{\rm T})
 =\lambda^3-2E_{\rm T}\lambda+C_3\,.
\end{align}
Thus, the three spectral parameters $\lambda_i$ introduced in
(\ref{tau1ansatz}) are precisely the eigenvalues of $\mathscr L_{\rm T}$ with 
\begin{align}
\sum_{i=1}^3\lambda_i=0,\qquad
E_{\rm T}=\frac14\sum_{i=1}^3\lambda_i^2,\qquad
C_3=-\lambda_1\lambda_2\lambda_3\,.
\end{align}
For the parametrization
$\lambda_1=\mu_1$, $\lambda_2=\mu_2$ and
$\lambda_3=-(\mu_1+\mu_2)$ used in the main text, these relations give
\begin{align}
E_{\rm T} =\frac12(\mu_1^2+\mu_1\mu_2+\mu_2^2)\,,\qquad 
C_3 =\mu_1\mu_2(\mu_1+\mu_2)\,.
\end{align}
This makes explicit the relation between the integration constants appearing
in the $\tau$-function solution and the conserved spectral data of the Toda flow. 
Together with the exponential representation of the
$\tau$-functions, the characteristic equation (\ref{TodaPoly}) implies that the $\tau$-functions obey 
the third-order linear equations 
\begin{align}
\label{}
\tau_1'''-2E_{\rm T}\tau_1'+C_3 \tau_1 =0 \,, \qquad 
\tau_2'''-2E_{\rm T}\tau_2'-C_3 \tau_2 =0 \,.
\end{align}
These equations remain valid for degenerate spectra, as can also be
verified directly from the Toda equations and the conserved quantities.

\subsection{Myers-Perry solution from the Toda system}
\label{app:MP}

Since the five-dimensional Myers-Perry solution 
with equal angular momenta falls into the cohomogeneity-one ansatz (\ref{metric0}), 
it is instructive to rederive it in the present framework. 

We solve (\ref{A1Toda}) and (\ref{A2Toda}) with the Hamiltonian constraint for the vanishing scalar charge $Q=0$. 
For the nonextremal black-hole branch, 
we choose the integration constants in (\ref{q0sol}) to be real, thus 
\begin{align}
\label{}
\tau_0(z)=\frac{\sinh (\mu_0 z)}{\mu_0} \,,  \qquad E_{\rm L}=\mu_0^2 >0\,, \qquad \mu_0>0 \,, 
\end{align}
where the other real constant $c_0$ has been set to vanish by using the shift degree of freedom $z\to z+{\rm const}$. 
For $q_i$, we choose the integration constants in (\ref{taui}) as 
\begin{align}
\label{}
\mu_1=2\mu_0\,, \qquad 
\mu_2=-2\mu_0 \,, \qquad 
c_1=\ln \left(\frac{\mu+\mu_0}{j^{2/3}}\right)\,, \qquad 
c_2= \ln \left(\frac{\mu-\mu_0}{j^{2/3}}\right)+i \pi \,, 
\end{align}
where $\mu$ is a real constant satisfying $j^2=\mu (\mu^2-\mu_0^2)$ and $\mu>\mu_0>0$. 
With this choice, the Hamiltonian constraint (\ref{constraint}) is satisfied, and 
the cubic integral vanishes identically $C_3 = 0$. 
Writing 
\begin{align}
\label{}
r^2+a^2 \equiv \frac 12 [\mu+\mu_0\coth (\mu_0 z)]\,, \qquad a\equiv \frac{j}{2\mu}\,, 
\end{align}
we have $\mu_0^2=\mu^2-4a^2 \mu$. 
The exterior region is covered by $0<z<\infty$, with $z\to0+$
corresponding to spatial infinity and $z\to\infty$ to the outer horizon. 
Finally, setting $\omega_0 =1/a$, we recover
\begin{align}
\label{MPmetric}
\D s^2=-\frac{\Delta_r}{\Sigma^2}\D t^2+\frac{r^2(r^2+a^2)}{\Delta_r }\D r^2
+\frac{\Sigma^2}{4(r^2+a^2)}\left(\sigma_3-\frac{2a\mu}{\Sigma^2}\D t\right)
^2+\frac 14 (r^2+a^2)(\sigma_1^2+\sigma_2^2 )\,,
\end{align}
where 
\begin{align}
\label{}
\Delta_r =(r^2+a^2)^2-\mu r^2 \,, \qquad \Sigma^2 = (r^2+a^2)^2+\mu a^2 \,. 
\end{align}
This is the five-dimensional vacuum Myers-Perry metric with equal angular momenta with
\begin{align}
\label{}
M=\frac{3\pi \mu}{8G}\,, \qquad 
J^{(\psi)}=\frac{\pi \mu a}{4G} \,. 
\end{align}
The event horizon corresponds to the largest root of $\Delta_r=0$ and is located at $r=r_+$, where  
\begin{align}
\label{}
r_+^2  =\frac 12 \left(\sqrt{\mu(\mu-4a^2)}+\mu-2a^2 \right)\,. 
\end{align}

Since the area $\ma A(r)$ of the constant $t$ and constant $r$ surfaces is given by 
\begin{align}
\label{maAMP}
\ma A(r)=2\pi^2 \sqrt{(r^2+a^2)[(r^2+a^2)^2+\mu a^2]}\,,
\end{align}
the area of the event horizon $A_{\rm H}=\ma A(r_+)$ is written as
\begin{align}
\label{}
A_{\rm H}=2\pi^2 r_+^{-1}(r_+^2+a^2)^2 \,. 
\end{align}
The surface gravity $\kappa$ and the angular velocity $\Omega_{\rm H}^{(\psi)}$ of the event horizon are given by 
\begin{align}
\label{}
\kappa =\frac{r_+^2-a^2}{r_+(r_+^2+a^2)}\,,\qquad 
\Omega_{\rm H}^{(\psi)}= \frac{2a}{r_+^2+a^2}\,. 
\end{align}
These quantities fulfill the first-law and the Smarr-formula 
\begin{align}
\label{}
\delta M=\frac{\kappa}{8\pi G}\delta A_{\rm H}+\Omega_{\rm H}^{(\psi)}\delta J^{(\psi)}\,, 
\qquad M=\frac 32 \left(\frac{\kappa A_{\rm H}}{8\pi G} +\Omega_{\rm H}^{(\psi)}J^{(\psi)}\right)\,.
\end{align}

The extremal limit is given by $\mu=4a^2$, for which 
\begin{align}
\label{mu0limit}
\mu_0=\sqrt{\mu(\mu-4a^2)}  \to 0 \,, \qquad  
\frac{G|J^{(\psi)}|}{A_{\rm H}}=\frac{a}{8\pi r_+} \to \frac 1{8\pi }\,.
\end{align}
It is worth commenting that 
the extremal solution corresponds to vanishing Liouville energy and 
saturates the upper bound in (\ref{JAHineq}).

The locus of the ergosurface is 
\begin{align}
\label{}
r_{\rm erg}=\sqrt{\mu-a^2} >r_+ \,. 
\end{align}
Contrary to the four-dimensional Kerr solution, the ergosurface does not touch the 
event horizon. 


\section{Curvature tensors}
\label{sec:curvature}

Although curvature invariants do not take particularly illuminating forms
when expressed directly in terms of the metric functions in (\ref{metric}),
they simplify considerably in the $\tau$-function representation
(\ref{metrictau}).

Before discussing the curvature invariants, it is useful to define
\begin{align}
\tilde\tau_1(z)
\equiv \frac{\tau_1(z)}{\tau_0(z)}=\frac{2m\,\tau_1(z)}{\cos(2mz)}\,,
\end{align}
obeying 
\begin{align}
\label{}
\ti \tau_1'{}^2-\ti \tau_1 \ti \tau_1 '' =\frac{\tau_2-\ti \tau_1^2}{\tau_0^2 } \,.
\end{align}
Since $\tau_1$ has a simple zero at each endpoint [see below (\ref{eq:TauEndpointLimits})], the zero of $\tau_1$
is precisely cancelled by the simple zero of $\cos(2mz)$.  Thus,
$\tilde\tau_1$ extends smoothly to both endpoints, with
\begin{align}
\tilde\tau_1(\pm z_0)
=\mp\tau_1'(\pm z_0)
=\sqrt{\tau_2(\pm z_0)}>0\,.
\end{align}
In particular, $\tilde\tau_1$ is finite and strictly positive throughout the
closed interval $-z_0\leq z\leq z_0$.

The scalar curvature then takes the simple and suggestive form
\begin{align}
R
&=
\frac{8\tau_0(z)^3\sqrt{\tau_2(z_0)}}{\ti \tau_1(z)}
(2E_{\rm L}-E_{\rm T})\,,
\end{align}
where $E_{\rm L}$ and $E_{\rm T}$ are the quadratic conserved quantities associated
with the Liouville and Toda sectors, respectively.
Since $\tilde\tau_1$ remains finite and nonzero at the endpoints,
the scalar curvature remains finite in the interior and vanishes $O((z\mp z_0)^3)$ as either asymptotic end is
approached $z\to \pm z_0$.

For the Kretschmann invariant, let us further define
\begin{align}
\Xi\equiv\frac{\tilde\tau_1'}{\tilde\tau_1}\,,
\qquad 
\Theta
\equiv-\frac{\tilde\tau_1''}{\tilde\tau_1}+\frac{2E_{\rm T}-E_{\rm L}}{3}\,.
\end{align}
Since $\tilde\tau_1$ is smooth and nonvanishing on the closed
interval, its derivatives are also finite there. Thus, $\Xi$ and
$\Theta$ remain finite for $-z_0\leq z\leq z_0$.
The Kretschmann invariant can then be written as
\begin{align}
R_{\mu\nu\rho\sigma}R^{\mu\nu\rho\sigma}
=&\,
\frac{32 \tau_2(z_0)\tau_0(z)^4}{\tilde\tau_1(z)^2}\,
\Biggl\{
36\Xi^2\tau_0'{}^2 -12 \left[\Xi(3\Theta+8\Xi^2-E_{\rm T})-C_3\right]\tau_0\tau_0'
\notag  \\
&+\tau_0^2 \left[3\Theta(7\Theta+32\Xi^2-2E_{\rm T}) +16 \Xi^2 (3\Xi^2-2E_{\rm L}+E_{\rm T})
-24 C_3 \Xi+\frac{11}3 (2E_{\rm L}-E_{\rm T})^2-3 E_{\rm T}^2\right]
\Biggr\}\,.
\label{Kret}
\end{align}
All quantities appearing inside the braces are finite at the
endpoints, since $\tau_0(\pm z_0)=0$ and
$\tau_0'(\pm z_0)=\mp1$, while $\Xi$ and $\Theta$ remain finite.
It therefore follows directly from (\ref{Kret}) that
the Kretschmann invariant behaves as
$O((z\mp z_0)^4)$ as $z\to\pm z_0$.
Since  the present solution fulfills  
\begin{align}
\label{}
R^{\mu\nu}R_{\mu\nu}=R^2 \,,  
\end{align}
we obtain
\begin{align}
C_{\mu\nu\rho\sigma}C^{\mu\nu\rho\sigma}
=R_{\mu\nu\rho\sigma}R^{\mu\nu\rho\sigma}-\frac76 R^2\,.
\end{align}
Thus, the Weyl invariant likewise vanishes 
$O((z\mp z_0)^4)$ as $z\to\pm z_0$.

\section{Rotating Ellis-Gibbons limit: $\nu\to 1$ \& $\beta\to \pm \infty$}
\label{app:EllisGibbons}

In section~\ref{sec:MP}, we discussed a scaling limit, which brings the metric 
to the exterior region of the Myers-Perry metric with equal angular momenta. 
In this appendix, we inspect another scaling limit in which the scalar field survives. 

Defining ($\hat m $, $\alpha$, $v$) by 
\begin{align}
\label{}
\hat m = m\beta >0 \,,  \qquad 
\nu= 1+\frac{\alpha }{\pi(1+3\alpha )\beta }\,, \qquad v= \beta u \,,
\end{align}
we take the limit $\beta \to +\infty$, while $\hat m$, $\alpha(>0)$, $v$ are kept fixed. 
In particular, $\nu\to1$ in this limit.

Then, we have 
\begin{align}
\label{}
D_- \to \frac{1}{1+3\alpha } >0 \,, \qquad 
 e^{-3\pi \beta}D_+ \to 1 \,, \qquad 
\beta k\to \alpha \,,
\end{align}
and the limiting solution reduces to 
\begin{align}
\label{EG}
\D s^2 =&\, -\frac{e^{-2v}}{H_2(v)}\D t^2 +\frac{\hat m e^{v}}{4v}\left[
H_1(v) \left(\frac{\D v^2}{v^2}+\sigma_1^2+\sigma_2^2\right)
+\frac{H_2(v)}{H_1(v)}\left(\sigma_3 -\omega(v)\D t \right)^2
\right] \,, 
\\
\phi=&\, \pm \sqrt{\frac 32} v \,, 
\end{align}
where 
\begin{align}
\label{}
W(v)=&\, 1-(1+3v)e^{-3v} \,, 
\\
H_1(v)=&\,  1+3\alpha -\alpha \frac{1-e^{-3v}}{v}\,, 
\\
H_2(v)=&\,  1+\alpha \left[3(2+3\alpha)W(v)+9 v e^{-3v} \right]\,, 
\\
\omega(v)=&\, \frac{2j}{9\hat m^2}\frac{W(v)}{H_2(v)}\,,
\end{align}
and 
\begin{align}
\label{}
j=9 \hat m^{3/2}\sqrt{\alpha }(1+3\alpha)\,. 
\end{align}
Here, we have used the same symbols for $v$-dependent metric functions,
in order to prevent the inflation of notations.

To express this limit in the Toda coordinate, we first shift its origin:
\begin{align}
\label{EGzshift}
\hat z=z-\frac{\pi}{4m}\,, \qquad
v=-2\hat m\hat z\,.
\end{align}
The limit is therefore taken at fixed $\hat z<0$, rather than at fixed $z$.
The Liouville and Toda variables are also well-defined in this limit
\begin{align}
\label{}
e^{q_0(\hat z)}={}&\frac{2\hat m}{v}
=-\frac{1}{\hat z}\,,\qquad 
\tau_1(\hat z)=
\frac{v e^v H_1(v)}{4\hat m j^{2/3}}\,, \qquad
\tau_2(\hat z)=
\frac{e^{2v}H_2(v)}{4j^{4/3}}\,.
\end{align}
The Liouville energy of the original wormhole family,
$E_{\rm L}=-4m^2=-4\hat m^2/\beta^2$, tends to zero.
Evaluating the conserved quantities of the limiting solution gives
\begin{align}
\label{EGintegrals}
E_{\rm L}=0\,, \qquad
E_{\rm T}=Q^2=6\hat m^2\,, \qquad
C_3=-16\hat m^3\,.
\end{align}
Thus, the vanishing Liouville energy results from the cancellation
between the Toda energy and the phantom-scalar contribution to the
Hamiltonian constraint.

The static limit is obtained by taking $\alpha\to0$, for which
$H_1=H_2=1$ and $j=\omega=0$.
Introducing $r^2=\hat m/v$, the solution  reduces to
\begin{align}
\label{EGstatic}
\D s^2={}&
-e^{-2\hat m/r^2}\D t^2
+e^{\hat m/r^2}\left(\D r^2+\frac{r^2}{4}(\sigma_1^2+\sigma_2^2+\sigma_3^2)\right)\,,
\qquad 
\phi= \pm\sqrt{\frac32}\,\frac{\hat m}{r^2}\,. 
\end{align}
This is the five-dimensional version \cite{Martinez:2020hjm}
of the Ellis-Gibbons solution~\cite{Ellis:1973yv,Gibbons:2003yj}, establishing
the interpretation of (\ref{EG}) as its rotating generalization.

The corresponding limit with $\beta\to-\infty$ yields the same
rotating Ellis-Gibbons family after interchanging the two asymptotic
ends and adopting canonically normalized, nonrotating coordinates at
the opposite end. We therefore restrict the explicit analysis to
$\beta\to+\infty$.

Let us briefly examine the geometry of the limiting solution.
For $\alpha>0$ and $v>0$, the metric functions satisfy
\begin{align}
\label{EGmonotonic}
\frac{\D H_1}{\D v}
=\frac{\alpha W(v)}{v^2}>0\,, \qquad
\frac{\D H_2}{\D v}
=9\alpha e^{-3v}
\left[1+3(1+3\alpha)v\right]>0\,.
\end{align}
Since $H_1(0)=H_2(0)=1$ and
$H_1(\infty)=1+3\alpha$,
$H_2(\infty)=(1+3\alpha)^2$, it follows that
\begin{align}
\label{EGpositivity}
1<H_1(v)<1+3\alpha\,, \qquad
1<H_2(v)<(1+3\alpha)^2
\qquad (0<v<\infty)\,.
\end{align}
The metric is consequently smooth and Lorentzian at every finite
positive $v$.
The limit $v\to0$ is an asymptotically flat end, as is manifest in the
coordinate $r^2=\hat m/v$. In the opposite limit, all curvature invariants vanish
sufficiently rapidly. For instance, 
\begin{align}
\label{EGRicci}
R=-\frac{6v^3e^{-v}}{\hat m H_1(v)} \to 0
\qquad (v\to\infty)\,.
\end{align}

One might then be tempted to expect that this limiting geometry also describes a 
traversable wormhole. However, this spacetime is not completely regular. 
To see this, consider the affinely parametrized null geodesics
\begin{align}
\label{}
\ell_\mu \D x^\mu = - \D t+\frac{\sqrt{\hat me^{3v}H_1H_2}}{2v^{3/2}}\D v \,, 
\end{align}
along which $\D v/\D \lambda =2 v^{3/2}e^{v/2}\sqrt{H_2/(\hat m H_1)}$, thus 
\begin{align}
\label{}
\lambda_*-\lambda=\int^\infty_v \frac{e^{-v/2}}{2v^{3/2}}\sqrt{\frac{\hat m H_1(v)}{H_2(v)}}\D v < \infty \,. 
\end{align}
This demonstrates that the $v=+\infty$ surface can be reached within a finite affine parameter. 
Thus, the $v=+\infty $ surface is not asymptotic infinity. Moreover, 
the parallelly propagated Ricci component diverges as $v\to \infty $
\begin{align}
\label{}
R_{\mu\nu}\ell^\mu \ell^\nu =-\frac{6v^3 e^v H_2(v)}{\hat m H_1(v)} \to -\infty \,. 
\end{align}
Thus, the $v=+\infty $ surface represents a parallelly propagated curvature singularity.
It follows that the limiting spacetime is geodesically incomplete, and
$v\to\infty$ does not represent a second regular asymptotically flat end. 
This property is also seen in the static counterpart \cite{Martinez:2020hjm}.

Although the present solution and the extremal Myers-Perry geometry both have vanishing Liouville energy $E_{\rm L}=0$, they correspond to distinct Toda data.  For the present rotating Ellis-Gibbons limit, the conserved quantities are given by (\ref{EGintegrals}), 
whereas the extremal Myers-Perry solution has $E_{\rm T}=C_3=0$.  Equivalently, the characteristic polynomial of the Toda Lax matrix factorizes as
\begin{align}
P_{\rm EG}(\lambda)
 =\lambda^3-2E_{\rm T}\lambda+C_3
 =(\lambda-4\hat m)(\lambda+2\hat m)^2\,,
\end{align}
while $P_{\rm MP}(\lambda)=\lambda^3$ in the extremal Myers-Perry limit.  Thus, although both solutions lie on the vanishing-discriminant locus $\Delta \equiv 32E_{\rm T}^3-27C_3^2 =0$ of the Toda characteristic polynomial, they correspond to different degenerations of the Lax spectrum.


\section{Cohomogeneity-one reduction in odd dimensions}
\label{app:odd}

It is natural to ask whether the Toda formulation developed in five
dimensions extends to higher odd dimensions.
To investigate this question, we consider the Einstein-scalar system
in $D=p+2+n$ dimensions, where $n=2n'+1\ge 3$ is odd:
\begin{align}
\label{}
\ma L^{(D)}=R^{(D)}-\epsilon  (\nabla \phi)^2\,. 
\end{align}
We adopt the following Kaluza-Klein reduction ansatz:
\begin{align}
\label{}
\D s_D^2=X^2 \left\{g_{ab}^{(p+2)}(y)\D y^a \D y^b + Y^{-2(n-1)}\left[\D \tau+A_a(y) \D y^a+A^{\rm FS}_I(w)\D w^I \right]^2
+Y^2 \gamma_{IJ}(w)\D w^I \D w^J  \right\}\,. 
\end{align}
Here, $\gamma_{IJ}(w)$ is  the Fubini-Study metric on $\mathbb{CP}^{n'} $ normalized by 
${}^{(\gamma)}R_{IJ}=(n+1)\gamma_{IJ}$, and $J^{\rm FS}=\frac 12 \D A^{\rm FS}$ is its K\"ahler form.  
The fields $X$, $Y$, $\phi$, and $A_a$, as well as
$g_{ab}^{(p+2)}$, depend only on the coordinates $y^a$.
The coordinate $\tau$ parametrizes the ${\rm U}(1)$ fibre over $\mathbb{CP}^{n'}$. 
With the standard period $2\pi$ for $\tau$, the angular manifold is
the Hopf bundle ${\rm U}(1)\to S^n \to \mathbb{CP}^{n'}$. 
The angular sector therefore admits a ${\rm U}(n'+1)$ isometry.

Writing $F=\D A$,  the $(p+2)$-dimensional effective Lagrangian is~\cite{Ma:2020kwc}
\begin{align}
\label{}
L^{(p+2)}= &\sqrt{-g^{(p+2)}} X^{ D-2} \Bigl\{R^{(p+2)}-\epsilon (\nabla \phi)^2 +( D-1)(D-2)(\nabla \ln X)^2
-n(n-1)(\nabla \ln Y)^2\notag \\
&- Y^{-2}\left[(n-1)Y^{-2n}-n^2+1\right]-\frac 14 Y^{-2(n-1)}F_{ab}F^{ab}\Bigr\}\,. 
\end{align}
All contractions and covariant derivatives are taken with respect to $g_{ab}^{(p+2)}$.

We now specialize to the $p=0$ ($D=2+n$) case, for which the equation for $F_{ab}$ can be 
solved as 
$F_{ab}=2 j \epsilon_{ab}^{(2)}X^{-n}Y^{2(n-1)}$, where $j$ is a constant and 
$\epsilon_{ab}^{(2)}$ is the volume form of $g_{ab}^{(2)}$. 
Performing the Routh reduction gives 
\begin{align}
\label{}
L^{(2)}= &\sqrt{-g^{(2)}} \Biggl\{X^{n} \Bigl(R^{(2)}-\epsilon (\nabla \phi)^2 +n(n+1)(\nabla \ln X)^2
-n(n-1)(\nabla \ln Y)^2\notag \\
&- Y^{-2}\left[(n-1)Y^{-2n}-n^2+1\right]\Bigr)-2X^{-n}Y^{2(n-1)}j^2 \Biggr\}\,. 
\end{align}

We next impose stationarity and take all remaining functions to
depend only on a radial coordinate $z$, with $y^a=(t,z)$.
Parametrizing the two-dimensional metric and the scalar fields $X$, $Y$ as
\begin{align}
\label{}
\D s_2^2=X^{-2} \left(-e^{2A(z)}\D t^2+N^2(z)\D z^2 \right)\,, \qquad 
X=e^{[(n-1)B+C]/n}\,, \qquad 
Y=e^{(B-C)/n} \,,
\end{align}
we obtain a stationary cohomogeneity-one ansatz 
\begin{align}
\label{}
\D s^2= -e^{2A(z)}\D t^2+N^2(z)\D z^2+e^{2C} \left[\D \tau+A_a(y) \D y^a+A^{\rm FS}_I(w)\D w^I \right]^2
+e^{2B} \gamma_{IJ}(w)\D w^I \D w^J\,. 
\end{align}
For $n=3$, we have $ \gamma_{IJ}(w)\D w^I \D w^J=(\sigma_1^2+\sigma_2^2)/4$, yielding 
$A^{\rm FS}=-\cos\vartheta \D \varphi /2$ and $A_t =\omega/2 $. Thus, 
$\tau=-\psi/2$ recovers the result in the body of text.

The scalar equation can be solved as 
 $\phi'=\epsilon e^{-A-(n-1)B-C}N Q$, where $Q$ is a constant. 
Performing a further Routh reduction at fixed $Q$ and discarding
total derivatives,  the following effective Lagrangian is obtained 
\begin{align}
\label{}
\tilde L=&\, \frac {e^{A+(n-1)B+C}}{N}\left[ 2A'\left((n-1)B'+C'\right)+(n-1)B'\left((n-2)B'+2C'\right)\right]
\notag \\
&+(n-1)N  e^{A+(n-1)B+C}\left((n+1)e^{-2B}-e^{2(C-2B)}\right)+N \left(-2 e^{A-(n-1)B-3C}j^2 + \epsilon e^{-A-(n-1)B-C}Q^2 \right)\,.
\label{Lag1D2higher}
\end{align}
Working in a gauge $N=e^{A+(n-1)B+C}$, we define
\begin{align}
\label{}
q_0&\,= (n-2)\left[B+\frac{2}{n-1}(A+C)+\frac{(n-3)\ln (64j^2)-(n-9)\ln [(n+1)(n-2)]}{2n(n-1)}\right]\,, 
\\
q_1&\, = \frac{2}{n-1}\left[2(A+C)+\frac{(n-2)\ln (64j^2)-2(n-3)\ln[(n+1)(n-2)]}{n}\right]\,, 
\\
q_2&\, =2A+\frac{n-1}{n}\ln (8j^2)+\frac{n-3}{n}\ln \left(\frac{2}{(n+1)(n-2)}\right) \,. 
\end{align}
Then, Euler-Lagrange equations for $q_i$ ($i=0,1,2$) are 
\begin{align}
\label{}
q_i''=\exp \left(\sum_{j=0}^2 \ma C_{ij}q_j\right)\,, \qquad 
\ma C_{ij}=\left(
\begin{array}{ccc}
 2     &  -\frac{n-3}{2} & 0  \\
\frac{2(n-3)}{n-2}   &  2 & -1 \\
0 & -\frac{n-1}{2} & 2   
\end{array}
\right) \,. 
\end{align}
When $n=3$, this system reduces to the Cartan matrix for 
$A_1\oplus A_2$ discussed in the main text. 
Thus, the resulting coupling matrix $\ma C$ does not correspond to the Cartan matrix 
for $n>3$. In particular, $\ma C_{10}=2(n-3)/(n-2)$ is positive, which 
does not satisfy the defining requirement of a Cartan matrix that its off-diagonal entries be non-positive.\footnote{
This obstruction cannot be removed by an invertible affine
redefinition of the variables that preserves a single exponential
term on the right-hand side of each equation.
Since the three exponential terms are linearly independent, such a
redefinition can only rescale and permute the variables, apart from constant shifts.
Under an independent rescaling $q_i=s_i\ti q_i+c_i$, the
exponent matrix becomes
$\ti {\ma C}_{ij}=\ma C_{ij}s_j$, 
$\ti{\ma C}_{ii}=2s_i$.  
Thus, imposing the Cartan normalization $\ti{\ma  C}_{ii}=2$
fixes $s_i=1$, while constant shifts affect only the prefactors of the exponential terms.
The positive off-diagonal entry therefore cannot be eliminated within this class of transformations.}
Although the $(q_1,q_2)$ block takes the Cartan form of $B_2\simeq C_2$ for $n=5$ and $G_2$ for $n=7$,
its coupling to $q_0$ prevents it from defining a closed two-variable Toda subsystem.
Consequently, the Cartan-type Toda integration used in five dimensions does not extend directly to the full
higher-dimensional system. 

The Cartan-type Toda structure identified above is therefore special
to five dimensions. This does not preclude the existence of explicitly integrable sectors
in higher dimensions. Indeed, the vacuum equations admit the equal-angular-momenta
Myers-Perry family for arbitrary odd $D$.
The existence of this family, however, does not by itself establish
integrability of the general reduced system. 

The distinction from the classification based on the closed conformal Killing-Yano tensors is worth emphasizing. 
In the work of Houri, Oota and Yasui \cite{Houri:2008th}, the metric is constrained by the existence of a closed conformal Killing-Yano two-form. 
For a nondegenerate tensor, this assumption leads to the canonical Kerr-NUT-(A)dS form, and the Einstein equations determine the corresponding structure functions. 
When some of the eigenvalues of the conformal Killing-Yano tensor become constant, the canonical coordinates degenerate and the metric instead takes a generalized form involving K\"ahler base spaces. 
In particular, by choosing the K\"ahler base to be the Fubini-Study metric on $\mathbb{CP}^{n'}$, one recovers the odd-dimensional Myers-Perry solution with equal angular momenta. 
Thus, equal-spin solutions are naturally incorporated in the degenerate conformal Killing-Yano classification, but the existence of the hidden symmetry is imposed from the outset. 

The present approach is different in that no conformal Killing-Yano tensor or separability structure is assumed. 
We start directly from the most general stationary cohomogeneity-one ansatz compatible with a ${\rm U}(1)$ fibration over $\mathbb{CP}^{n'}$ and reduce the field equations to a one-dimensional mechanical system. 
In five dimensions and in the absence of a cosmological constant, this reduced system turns out to be governed by an $A_1\oplus A_2$ Toda system, from which the equal-angular-momentum Myers-Perry solution is recovered as a particular vacuum solution. 
The integrable structure therefore emerges from the reduced Einstein equations themselves, rather than being inferred from a prescribed hidden symmetry. 
This formulation can therefore be applied directly in the presence of matter fields, 
such as the phantom scalar considered here, without assuming that a closed conformal Killing-Yano tensor survives.




\begin{thebibliography}{99}


\bibitem{Stephani:2003tm}
  H.~Stephani, D.~Kramer, M.~A.~H.~MacCallum, C.~Hoenselaers and E.~Herlt,
  ``Exact solutions of Einstein's field equations,''
  Cambridge, UK: Univ. Pr. (2003) 701 P.


\bibitem{Emparan:2008eg}
R.~Emparan and H.~S.~Reall,
Living Rev. Rel. \textbf{11}, 6 (2008)
doi:10.12942/lrr-2008-6
[arXiv:0801.3471 [hep-th]].


\bibitem{Iguchi:2011qi}
H.~Iguchi, K.~Izumi and T.~Mishima,
Prog. Theor. Phys. Suppl. \textbf{189}, 93-125 (2011)
doi:10.1143/PTPS.189.93
[arXiv:1106.0387 [gr-qc]].
  

  
\bibitem{Breitenlohner:1987dg}
P.~Breitenlohner, D.~Maison and G.~W.~Gibbons,
Commun. Math. Phys. \textbf{120}, 295 (1988)
doi:10.1007/BF01217967

\bibitem{Morris:1988cz}
M.~S.~Morris and K.~S.~Thorne,
Am. J. Phys. \textbf{56}, 395-412 (1988)
doi:10.1119/1.15620

\bibitem{Morris:1988tu}
M.~S.~Morris, K.~S.~Thorne and U.~Yurtsever,
Phys. Rev. Lett. \textbf{61}, 1446-1449 (1988)
doi:10.1103/PhysRevLett.61.1446


\bibitem{Teo:1998dp}
E.~Teo,
Phys. Rev. D \textbf{58}, 024014 (1998)
doi:10.1103/PhysRevD.58.024014
[arXiv:gr-qc/9803098 [gr-qc]].



\bibitem{Ellis:1973yv}
H.~G.~Ellis,
J. Math. Phys. \textbf{14}, 104-118 (1973)
doi:10.1063/1.1666161

\bibitem{Bronnikov:1973fh}
K.~A.~Bronnikov,
Acta Phys. Polon. B \textbf{4}, 251-266 (1973)



\bibitem{Kleihaus:2014dla}
B.~Kleihaus and J.~Kunz,
Phys. Rev. D \textbf{90}, 121503 (2014)
doi:10.1103/PhysRevD.90.121503
[arXiv:1409.1503 [gr-qc]].

\bibitem{Chew:2016epf}
X.~Y.~Chew, B.~Kleihaus and J.~Kunz,
Phys. Rev. D \textbf{94}, no.10, 104031 (2016)
doi:10.1103/PhysRevD.94.104031
[arXiv:1608.05253 [gr-qc]].



\bibitem{Chew:2019lsa}
X.~Y.~Chew, V.~Dzhunushaliev, V.~Folomeev, B.~Kleihaus and J.~Kunz,
Phys. Rev. D \textbf{100}, no.4, 044019 (2019)
doi:10.1103/PhysRevD.100.044019
[arXiv:1906.08742 [gr-qc]].


\bibitem{Matos:2005uh}
T.~Matos and D.~Nunez,
Class. Quant. Grav. \textbf{23}, 4485-4496 (2006)
doi:10.1088/0264-9381/23/13/012
[arXiv:gr-qc/0508117 [gr-qc]].


\bibitem{Klemm:2015uba}
D.~Klemm, M.~Nozawa and M.~Rabbiosi,
Class. Quant. Grav. \textbf{32}, no.20, 205008 (2015)
doi:10.1088/0264-9381/32/20/205008
[arXiv:1506.09017 [hep-th]].


\bibitem{Nozawa:2026jls}
M.~Nozawa, T.~Harada and K.~Nakashi,
[arXiv:2610.05656 [gr-qc]].

\bibitem{Myers:1986un}
R.~C.~Myers and M.~J.~Perry,
Annals Phys. \textbf{172}, 304 (1986)
doi:10.1016/0003-4916(86)90186-7

\bibitem{Gibbons:2004uw}
G.~W.~Gibbons, H.~Lu, D.~N.~Page and C.~N.~Pope,
J. Geom. Phys. \textbf{53}, 49-73 (2005)
doi:10.1016/j.geomphys.2004.05.001
[arXiv:hep-th/0404008 [hep-th]].







\bibitem{Dzhunushaliev:2013jja}
V.~Dzhunushaliev, V.~Folomeev, B.~Kleihaus, J.~Kunz and E.~Radu,
Phys. Rev. D \textbf{88}, 124028 (2013)
doi:10.1103/PhysRevD.88.124028
[arXiv:1309.2448 [gr-qc]].

\bibitem{Uemichi:2026dzb}
K.~Uemichi, Y.~Koga, D.~Saito, C.~M.~Yoo and D.~Yoshida,
Phys. Rev. D \textbf{113}, no.12, 124054 (2026)
doi:10.1103/w1f7-r1x6
[arXiv:2603.12748 [gr-qc]].


\bibitem{Blazquez-Salcedo:2026wbi}
J.~L.~Bl{\'a}zquez-Salcedo, L.~M.~Gonz{\'a}lez-Romero, F.~S.~Khoo, B.~Kleihaus, J.~Kunz and P.~Navarro Moreno,
[arXiv:2610.38659 [gr-qc]].




\bibitem{Torii:2013xba}
T.~Torii and H.~a.~Shinkai,
Phys. Rev. D \textbf{88}, 064027 (2013)
doi:10.1103/PhysRevD.88.064027
[arXiv:1309.2058 [gr-qc]].


\bibitem{Martinez:2020hjm}
C.~Martinez and M.~Nozawa,
Phys. Rev. D \textbf{103}, no.2, 024003 (2021)
doi:10.1103/PhysRevD.103.024003
[arXiv:2010.05183 [gr-qc]].


\bibitem{Huang:2019arj}
H.~Huang and J.~Yang,
Phys. Rev. D \textbf{100}, no.12, 124063 (2019)
doi:10.1103/PhysRevD.100.124063
[arXiv:1909.04603 [gr-qc]].


\bibitem{Nozawa:2020wet}
M.~Nozawa,
Phys. Rev. D \textbf{103}, no.2, 024004 (2021)
doi:10.1103/PhysRevD.103.024004
[arXiv:2010.07560 [gr-qc]].




\bibitem{Huang:2020qmn}
H.~Huang, H.~L{\"u} and J.~Yang,
Class. Quant. Grav. \textbf{39}, no.18, 185009 (2022)
doi:10.1088/1361-6382/ac8266
[arXiv:2010.00197 [gr-qc]].




\bibitem{Nozawa:2020gzz}
M.~Nozawa,
Phys. Rev. D \textbf{103}, no.2, 024005 (2021)
doi:10.1103/PhysRevD.103.024005
[arXiv:2010.07561 [gr-qc]].

\bibitem{Blazquez-Salcedo:2020nsa}
J.~L.~Bl{\'a}zquez-Salcedo, X.~Y.~Chew, J.~Kunz and D.~H.~Yeom,
Eur. Phys. J. C \textbf{81}, no.9, 858 (2021)
doi:10.1140/epjc/s10052-021-09645-0
[arXiv:2012.06213 [gr-qc]].


\bibitem{Nozawa:2023aep}
M.~Nozawa and T.~Torii,
Phys. Rev. D \textbf{108}, no.6, 064036 (2023)
doi:10.1103/PhysRevD.108.064036
[arXiv:2306.15198 [gr-qc]].


\bibitem{Tahamtan:2016fur}
T.~Tahamtan and O.~Svitek,
Phys. Rev. D \textbf{94}, no.6, 064031 (2016)
doi:10.1103/PhysRevD.94.064031
[arXiv:1603.07281 [gr-qc]].


\bibitem{Nozawa:2023boa}
M.~Nozawa and T.~Torii,
Class. Quant. Grav. \textbf{41}, no.6, 065016 (2024)
doi:10.1088/1361-6382/ad26ec
[arXiv:2310.05460 [hep-th]].




\bibitem{Bogoyavlensky}
O. I. Bogoyavlensky, 
Comm. Math. Phys. {\bf 51},  201 (1976)
doi: 10.1007/BF01617919


\bibitem{Kostant:1979qu}
B.~Kostant,
Adv. Math. \textbf{34}, 195-338 (1979)
doi:10.1016/0001-8708(79)90057-4

\bibitem{Olshanetsky}
M. A. Olshanetsky,
Physica D: Nonlinear Phenomena,
{\bf 3}, 118 (1981). 
doi:10.1016/0167-2789(81)90121-4


\bibitem{Anderson:1995sz}
A.~Anderson,
J. Math. Phys. \textbf{37}, 1349-1355 (1996)
doi:10.1063/1.531465
[arXiv:hep-th/9507092 [hep-th]].

\bibitem{Gibbons:1987ps}
G.~W.~Gibbons and K.~i.~Maeda,
Nucl. Phys. B \textbf{298}, 741-775 (1988)
doi:10.1016/0550-3213(88)90006-5


\bibitem{Monni:1995vu}
S.~Monni and M.~Cadoni,
Nucl. Phys. B \textbf{466}, 101-111 (1996)
doi:10.1016/0550-3213(96)00090-9
[arXiv:hep-th/9511067 [hep-th]].


\bibitem{Lu:1996hh}
H.~Lu, C.~N.~Pope and K.~W.~Xu,
Mod. Phys. Lett. A \textbf{11}, 1785-1796 (1996)
doi:10.1142/S0217732396001776
[arXiv:hep-th/9604058 [hep-th]].

\bibitem{Chemissany:2010ay}
W.~Chemissany, J.~Rosseel and T.~Van Riet,
Nucl. Phys. B \textbf{843}, 413-428 (2011)
doi:10.1016/j.nuclphysb.2010.10.006
[arXiv:1009.1487 [hep-th]].

\bibitem{Lu:2013uia}
H.~L{\"u} and W.~Yang,
Class. Quant. Grav. \textbf{30}, 235021 (2013)
doi:10.1088/0264-9381/30/23/235021
[arXiv:1307.2305 [hep-th]].


\bibitem{Lu:2026lkv}
H.~Lu, P.~Y.~Wu, Z.~H.~Wu and W.~Zhao,
JHEP \textbf{08}, 112 (2026)
doi:10.1007/JHEP08(2026)112
[arXiv:2604.14288 [hep-th]].


\bibitem{Yun:2026ure}
S.~Yun,
[arXiv:2609.23079 [hep-th]].






\bibitem{Heusler:1996jaf}
M.~Heusler,
1996,
ISBN 978-0-511-66139-6
doi:10.1017/cbo9780511661396



\bibitem{Minguzzi:2006sa}
E.~Minguzzi and M.~Sanchez,
[arXiv:gr-qc/0609119 [gr-qc]].



\bibitem{Graf:2017xym}
M.~Graf and E.~Ling,
Class. Quant. Grav. \textbf{35}, no.8, 087001 (2018)
doi:10.1088/1361-6382/aab259
[arXiv:1712.06504 [gr-qc]].





\bibitem{Gibbons:2013tqa}
G.~W.~Gibbons and N.~P.~Warner,
Class. Quant. Grav. \textbf{31}, 025016 (2014)
doi:10.1088/0264-9381/31/2/025016
[arXiv:1305.0957 [hep-th]].


\bibitem{Eperon:2016cdd}
F.~C.~Eperon, H.~S.~Reall and J.~E.~Santos,
JHEP \textbf{10}, 031 (2016)
doi:10.1007/JHEP10(2016)031
[arXiv:1607.06828 [hep-th]].










\bibitem{Nozawa:2025yac}
M.~Nozawa and T.~Torii,
Phys. Rev. D \textbf{113}, no.6, 064021 (2026)
doi:10.1103/62b3-kqcw
[arXiv:2510.06561 [gr-qc]].

\bibitem{Khoo:2024yeh}
F.~S.~Khoo, B.~Azad, J.~L.~Bl{\'a}zquez-Salcedo, L.~M.~Gonz{\'a}lez-Romero, B.~Kleihaus, J.~Kunz and F.~Navarro-L{\'e}rida,
Phys. Rev. D \textbf{109}, no.8, 084013 (2024)
doi:10.1103/PhysRevD.109.084013
[arXiv:2401.02898 [gr-qc]].


\bibitem{Murata:2007gv}
K.~Murata and J.~Soda,
Class. Quant. Grav. \textbf{25}, 035006 (2008)
doi:10.1088/0264-9381/25/3/035006
[arXiv:0710.0221 [hep-th]].





\bibitem{Nozawa:2005eu}
M.~Nozawa and K.~i.~Maeda,
Phys. Rev. D \textbf{71}, 084028 (2005)
doi:10.1103/PhysRevD.71.084028
[arXiv:hep-th/0502166 [hep-th]].





\bibitem{Olshanetsky:1981dk}
M.~A.~Olshanetsky and A.~M.~Perelomov,
Phys. Rept. \textbf{71}, 313 (1981)
doi:10.1016/0370-1573(81)90023-5

\bibitem{Gibbons:2003yj}
G.~W.~Gibbons,
[arXiv:hep-th/0302199 [hep-th]].



\bibitem{Ma:2020kwc}
L.~Ma and H.~L{\"u},
JHEP \textbf{03}, 226 (2021)
doi:10.1007/JHEP03(2021)226
[arXiv:2011.12971 [hep-th]].

\bibitem{Houri:2008th}
T.~Houri, T.~Oota and Y.~Yasui,
Phys. Lett. B \textbf{666}, 391-394 (2008)
doi:10.1016/j.physletb.2008.07.075
[arXiv:0805.0838 [hep-th]].


\end{thebibliography}
\end{document}